\documentclass[
10pt,
]{scrartcl}

\usepackage[utf8]{inputenc}
\usepackage[english]{babel}
\usepackage{amsfonts,amsmath,amssymb,bbm}
\usepackage{amsthm}
\usepackage{nicefrac}
\usepackage{pifont}
\usepackage{yfonts}
\usepackage{geometry} 
\usepackage{graphicx}
\usepackage{booktabs}  
\usepackage{array} 
\usepackage{paralist} 
\usepackage{verbatim} 
\usepackage{subfig} 
\usepackage[nottoc,notlof,notlot]{tocbibind} 
\usepackage[titles,subfigure]{tocloft}  
\usepackage{dsfont}
\usepackage{hyperref}

\newcommand{\gs}{{\mathrm{gs}}}
\newcommand{\hf}{{\mathrm{HF}}}

\newcommand{\vphi}{{\varphi}}           

\newcommand{\ol}{\overline}

\newcommand{\di}{\mathrm{d}}

\newcommand{\e}{\mathrm{e}}
\newcommand{\cA}{\mathcal{A}}
\newcommand{\cB}{\mathcal{B}}
\newcommand{\cC}{\mathcal{C}}
\newcommand{\cD}{\mathcal{D}}
\newcommand{\cE}{\mathcal{E}}
\newcommand{\cF}{\mathcal{F}}

\newcommand{\cK}{\mathcal{K}}
\newcommand{\cL}{\mathcal{L}}         
\newcommand{\cQ}{\mathcal{Q}}
\newcommand{\cR}{\mathcal{R}}
\newcommand{\cS}{\mathcal{S}}

\newcommand{\cW}{\mathcal{W}}

\newcommand{\cZ}{\mathcal{Z}}

\newcommand{\fH}{\mathfrak{H}}          
          
\newcommand{\fS}{\mathfrak{S}}

\newcommand{\fh}{\mathfrak{h}}

\newcommand{\rIm}{\mathop{\mathrm{Im}}}

\def\I{\text{i}}   

\renewcommand{\Re}{\operatorname{Re}}
\renewcommand{\Im}{\operatorname{Im}}

\newcommand{\Po}{\mathcal{P}}
\newcommand{\hgamma}{\widehat{\gamma}}
\newcommand{\hGamma}{\widehat{\Gamma}}
\newcommand{\tgamma}{\widetilde{\gamma}}
\newcommand{\Ex}{{\mathrm{Ex}}}

\newcommand{\dR}{\mathds{R}}
\newcommand{\1}{\mathds{1}}
\newcommand{\hh}{\mathds{h}}
\newcommand{\NN}{\mathds{N}}
\newcommand{\hNN}{\widehat{\mathds{N}}}
\newcommand{\VV}{\mathds{V}}
\newcommand{\dC}{\mathds{C}}

\newcommand{\Hf}{\mathds{H}}

\newcommand{\Laplace}{\it{\Delta}}
\newcommand{\vac}{\Omega}
\def\fc{{c}^{*}}
\def\fa{{c}}
\def\bc{{a}^{*}}
\def\ba{{a}}

\newcommand{\ra}{\rightarrow}

\newcommand{\tr}{\mathrm{tr}}
\def\abs#1{\left| #1 \right|}
\def\sn#1{\left\| #1 \right\|}

\def\bra#1{ \left< #1 \right|}
\def\ket#1{ \left| #1 \right>}

\newcommand{\eg}{e.g.}
\newcommand{\ie}{i.e.}

\newcommand\relphantom[1]{\mathrel{\phantom{#1}}}

\def\TSkip{\bigskip}
\def\Abstract{\begingroup\narrower
    \parskip=\medskipamount\parindent=0pt{\sc{Abstract}. }}
\def\EndAbstract{\par\endgroup\TSkip}
\long\def\MSC#1\EndMSC{\def\arg{#1}\ifx\arg\empty\relax\else
     {\par\narrower\noindent%
     2010 Mathematics Subject Classification: #1\par}\fi}
\long\def\KEY#1\EndKEY{\def\arg{#1}\ifx\arg\empty\relax\else
	{\par\narrower\noindent Keywords and Phrases: #1\par}\fi\TSkip}

\begin{document}
\theoremstyle{plain}
\newtheorem{theorem}{Theorem}[section]
\newtheorem{lemma}[theorem]{Lemma}
\newtheorem{corollary}[theorem]{Corollary}
\newtheorem{proposition}[theorem]{Proposition}
\theoremstyle{definition}
\newtheorem{definition}[theorem]{Definition}
\theoremstyle{remark}
\newtheorem{remark}[theorem]{Remark}

\hbadness=100


\title{Generalized One-Particle Density Matrices and Quasifree States}
\author{Volker Bach\footnote{Email: v.bach@tu-bs.de},
S\'ebastien Breteaux\footnote{Email: sebastien.breteaux@ens-cachan.org}\\
{\small{\em{Technische Universit\"at Braunschweig, Institut f\"ur Analysis und Algebra,}}}\\
{\small{\em{Rebenring 31, 38106 Braunschweig, Germany}}}\\
Hans Konrad Kn\"orr\footnote{Email: hanskonrad.knoerr@fernuni-hagen.de}\\
{\small{\em{FernUniversit\"at in Hagen, Fakult\"at f\"ur Mathematik und Informatik,}}}\\
{\small{\em{Lehrgebiet Angewandte Stochastik, 58084 Hagen, Germany}}}\\
Edmund Menge\footnote{Email: e.menge@tu-bs.de}\\
{\small{\em{Technische Universit\"at Braunschweig, Institut f\"ur Analysis und Algebra,}}}\\
{\small{\em{Rebenring 31, 38106 Braunschweig, Germany}}}}


\maketitle

\Abstract
In the spirit of the generalized one-particle density matrix for fermions, we introduce generalized one- and two-particle density matrices to state re\-pre\-sen\-ta\-bi\-li\-ty conditions up to second order for boson systems without assuming particle number-conservation. Furthermore, we show for both particle species that, for a semibounded Hamiltonian, the infimum of the variation of the energy functional w.r.t quasifree states coincides with the one of a variation over pure quasifree states. Moreover, it is proven for fermions that only pure quasifree states have a generalized 1-pdm that is a projection, and a similar statement for bosons.
\EndAbstract

\MSC
81V70
\EndMSC

\KEY
Bogoliubov--Hartree--Fock Theory, Hartree--Fock theory, re\-pre\-sen\-ta\-bi\-li\-ty, Quasifree States
\EndKEY

\tableofcontents


\section{Introduction}

The Rayleigh--Ritz variational principle for the ground state energy is
the starting point of many computations and approximations in quantum
chemistry. For a many-particle system whose dynamics is generated by a
Hamiltonian $\Hf$, it can be written as
\begin{align} \label{eq-vb-1}
E_{\gs} 
\ = \ 
\inf\Big\{ \tr_\cF( \rho \Hf ) \; \Big| 
\; \rho \geq 0, \ \tr_\cF(\rho) = 1 \Big\},
\end{align}
where $\rho$ varies over the density matrices on the Fock space $\cF^{\pm}
\equiv \cF^{\pm}[\fh]$ of the system ($E_{\gs}$ in \eqref{eq-vb-1} is
actually the \emph{total} ground state energy in the grand canonical
ensemble). A typical many-particle Hamiltonian is given as the sum
$\Hf = \hh + \VV$ of the second quantization $\hh$ of a one-particle
operator $h$ and the second quantization $\VV$ of a pair potential
$V$. Since $\hh$ is quadratic and $\VV$ is quartic in the field
operators, one can rewrite \eqref{eq-vb-1} in terms of the
one-particle density matrix $\gamma_\rho \in \cL_+^1(\fh)$ and
the two-particle density matrix $\Gamma_\rho \in \cL_+^1(\fh
\otimes \fh)$ of a given density matrix $\rho \in \cL_+^1(\cF)$ as
\begin{align*} 
E_{gs} 
\ = \ 
\inf\Big\{ \cE(\gamma_\rho, \Gamma_\rho) \; \Big| 
\; \rho \geq 0, \ \tr_\cF(\rho) = 1 \Big\},
\end{align*}
where the energy functional $\cE$ is defined by
\begin{align*} 
\cE(\gamma_{\rho}, \Gamma_{\rho})
\ := \ 
\tr_{\fh}( h \, \gamma_\rho ) \: + \:
\tfrac{1}{2} \tr_{\fh \otimes \fh}( h \, \Gamma_\rho ) .
\end{align*}
The computation of the ground state energy and the corresponding
ground state vector of a quantum mechanical many-particle system is a
complex, if not impossible, task and one resorts to approximation
methods. The Hartree--Fock approximation is one of the first
approximations that emerged from ground state computations in quantum
chemistry \cite{Ba1,BLS,BKM}. In
its original formulation, the Rayleigh--Ritz principle for the ground
state energy in terms of wave functions,
\begin{align*} 
E_{\gs} 
\ = \ 
\inf\Big\{ \big\langle \Psi, \Hf \Psi \big\rangle_{\cF} \Big|  
\ \Psi \in \cF , \ \|\Psi\|_{\cF}=1 \Big\}
\end{align*}
of a fermion system with Hamiltonian $\Hf$ is replaced by a variation
over Slater determinants,
\begin{align} \label{eq-vb-5}
E_{\hf} 
\ = \ 
\inf\Big\{ \big\langle \Phi, \Hf \Phi \big\rangle_{\cF} \Big| 
\ N \in \NN, \ \Phi = \vphi_1 \wedge \cdots \wedge \vphi_N , 
\ \langle \vphi_i, \vphi_j \rangle_{\fh} = \delta_{i,j} \Big\} ,
\end{align}
where the Hamiltonian $\Hf$ conserves the particle number, i.e., $[ \Hf ,
  \hNN ] = 0$, with $\hNN$ being the particle number operator. The
density matrix $\rho = |\Phi\rangle\langle\Phi|$ associated to a
Slater determinant is a pure, particle number-conserving, quasifree
state and \eqref{eq-vb-5} can be rewritten as
\begin{align*} 
E_{\hf} 
\ = \ 
\inf\Big\{ \tr_\cF( \rho \Hf ) \; \Big|
\ \text{$\rho$ is a pure, particle number-conserving, quasifree
  density matrix} \Big\}.
\end{align*}
Since the one-particle density matrix $\gamma$ of a fermion
Slater determinant $\Phi = \vphi_1 \wedge \cdots \wedge \vphi_N$ is
the rank-$N$ orthogonal projection onto $\mathrm{span}\{ \vphi_1,
\ldots, \vphi_N\}$ and its two-particle density matrix is
given as $\Gamma = (\1 - \Ex) (\gamma \otimes \gamma)$, the
Hartree--Fock energy can be written as
\begin{align*} 
E_{\hf} 
\ = \ 
\inf\Big\{ \cE \big( \gamma, (\1 - \Ex) (\gamma \otimes \gamma) \big) \Big| 
\; \gamma = \gamma^2 = \gamma^*, \ \tr_{\fh}(\gamma) < \infty \Big\}.
\end{align*}
In case of purely repulsive pair potentials $V$, Lieb's variational
principle \cite{Lie,Ba1,BLS} asserts that 
\begin{align*} 
E_{\hf} 
\ = \ 
\inf\Big\{ \cE \big(\gamma, (\1 - \Ex) (\gamma \otimes \gamma)\big) \Big| 
\; 0 \leq \gamma \leq \1, \ \tr_{\fh}(\gamma) = N \Big\}.
\end{align*}
Going back to a description on the Fock space, Lieb's variational
principle reads
\begin{align} \label{eq-vb-9}
E_{\hf} 
\ = \ &
\inf\Big\{ \tr_\cF( \rho \Hf ) \; \Big| 
\ \text{$\rho$ is a particle number-conserving, quasifree
  density matrix} \Big\},
\end{align}
i.e., it asserts that the pureness requirement of the quasifree
density matrix can be dropped. As was shown in
\cite{BLS}, the property $[\rho, \hNN ] = 0$ of
particle number conservation is also obsolete for repulsive pair
potentials $V$, and the Hartree--Fock energy $E_{\hf}$ agrees
with the Bogoliubov--Hartree--Fock energy $E_{\mathrm{BHF}}$ defined by
\begin{align} \label{eq-vb-10}
E_{\mathrm{BHF}} 
\ := \ 
\inf\Big\{ \tr_\cF( \rho \Hf ) \; \Big| 
\ \text{$\rho$ is a quasifree density matrix} \Big\}.
\end{align}
Our first main result is a generalization of Lieb's variational principle
\eqref{eq-vb-9} in several ways. Namely, we show that the infimum
in \eqref{eq-vb-10} is already obtained from a variation over 
\emph{pure} quasifree density matrices,
\begin{align*} 
E_{\mathrm{BHF}} 
\ = \ 
E_{\mathrm{BHF}}^{\mathrm{(pure)}} 
\ = \ 
\inf\Big\{ \tr_\cF( \rho \Hf ) \; \Big| 
\ \text{$\rho$ is a pure, quasifree density matrix} \Big\},
\end{align*}
under the mere assumption that $\Hf$ is bounded below. Neither
repulsiveness of the pair potential $V$ nor the form $\Hf = \hh + \VV$
or even the conservation of the particle number by $\Hf$ is
assumed. Furthermore, we show that $E_{\mathrm{BHF}} =
E_{\mathrm{BHF}}^{\mathrm{(pure)}}$ for both fermion and boson
systems. The precise formulation of this first result and its proof is
given in Theorem~\ref{thm: Variation}. Note that, especially for boson systems, it
is crucial that our result does not require the Hamiltonian to
conserve the particle number because for most physically interesting
models such an assumption would not be fulfilled.

The above result, i.e., Theorem~\ref{thm: Variation} brings pure, quasifree
density matrices $\rho$ into focus. These are fully characterized by
their generalized one-particle density matrix $\tgamma_\rho$ defined
in terms of their two-point correlation functions as
\begin{align*} 
\Big\langle f_1 \oplus f_2, 
\; \tgamma_\rho \, (g_1 \oplus g_2) \Big\rangle_{\fh \oplus \fh}
\ := \ 
\tr_\cF\Big( \rho \, \big[ a^*(g_1) + a(\ol{g}_2) \big] 
\, \big[ a(f_1) + a^*(\ol{f}_2) \big] \Big) ,
\end{align*}
where $\{ a^*(f), a(f) | \, f \in \fh \}$ are the usual boson or fermion
creation operators on $\cF^{\pm}$ fulfilling the canonical commutation or
anticommutation relations, respectively, on $\cF^{\pm}$, with $a(f)$
annihilating the vacuum and $f \mapsto J(f) =: \bar{f}$ being (a fixed
antilinear involution on $\fh$ which we refer to as) the complex
conjugation. Here, we implicitly assume $\tr_\cF ( \rho \, a(f) ) = 0$,
for all $f \in \fh$, i.e., that $\rho$ is \emph{centered}. This
assumption is irrelevant for fermion systems and made without loss of
generality for boson systems, as is explained below. The higher
correlation of the (centered) quasifree density matrix $\rho$ can be
computed from sums over products of the two-point correlation
function, i.e., in terms of $\tgamma_\rho$, using Wick's theorem. It
is well-known \cite{BLS,So1} that, as a $2
\times 2$ matrix with operator-valued entries, $\tgamma_\rho$ can be
written as
\begin{align} \label{eq-vb-13}
\tgamma_\rho 
\ = \ 
\begin{pmatrix} \gamma_\rho & \alpha_\rho 
\\ \ol{\alpha}_\rho & \1 \pm \ol{\gamma}_\rho
\end{pmatrix}, 
\qquad
\gamma_\rho = \gamma_\rho^* , \quad  \alpha_\rho = \pm \alpha_\rho^T ,
\end{align}
where ``$+$'' holds for boson and ``$-$'' for fermion systems, and
$\ol{A} := J A J$ denote the complex conjugate and $A^T := \ol{A}^*$
the transpose of a bounded operator $A \in \cB(\fh)$. 
It is easy to check that 
\begin{align} \label{eq-vb-14}
\tgamma_\rho \ \geq \ 0, \ \ \ \; & \qquad \text{for boson systems,}  
\\ \label{eq-vb-15}
0 \ \leq \ \tgamma_\rho \ \leq \ \1, & \qquad \text{for fermion systems.}  
\end{align}
We restrict our attention to density matrices with finite particle
number expectation, for which
\begin{align*} 
\tr_{\fh}(\gamma_\rho) 
\ = \ 
\tr_\cF( \rho \, \hNN ) 
\ < \ \infty.
\end{align*}
In this case, it is well-known that the converse of \eqref{eq-vb-14} and
\eqref{eq-vb-15} holds true in the sense that, given
$\tgamma$ as in \eqref{eq-vb-13}, with
$\gamma = \gamma^* \in \cL^1(\fh)$, $\gamma \geq 0$, $\alpha = \pm \alpha^T$ 
and obeying
\begin{align*} 
\tgamma \ \geq \ 0, \ \ \ \; & \qquad \text{for boson systems,}  
\\ 
0 \ \leq \ \tgamma \ \leq \ \1, & \qquad \text{for fermion systems,}  
\end{align*}
there exists a centered quasifree density matrix $\rho \in
\cL_+^1(\cF)$ such that
\begin{align*} 
\tgamma \ = \ \tgamma_\rho . 
\end{align*}
It is furthermore well-known \cite{BLS,So1} that, 
if $\rho$ is a pure, quasifree density matrix, then 
\begin{align} \label{eq-vb-20}
\tgamma_\rho \ = \ \tgamma_\rho^2 , \qquad &
\quad \text{for fermion systems and}
\\ \label{eq-vb-21} 
\tgamma_\rho \ = \ - \tgamma_\rho \: \cS \: \tgamma_\rho, &
\quad \text{for boson systems,}
\end{align}
where
\begin{align*} 
\cS \ = \ 
\begin{pmatrix} \1 & 0 \\ 0 & -\1 \\ \end{pmatrix}.
\end{align*}
Our second main result is the converse statement: If the generalized
one-particle density matrix $\tgamma_\rho$ of a density matrix $\rho$
fulfills \eqref{eq-vb-20}, in the fermion case, or \eqref{eq-vb-21}, in
the boson case, then the density matrix is a pure, quasifree
state. Note that the quasifreeness of $\rho$ is asserted, not assumed.
The precise formulation of this result is given in Theorem~\ref{thm: relation}.

As our third result, we derive representability conditions on the
two-particle density matrix $\Gamma_\rho$ of a boson density matrix
$\rho$. Similar to G-, P-, and Q-Conditions for fermion reduced density matrices,
these conditions follow from the positivity 
\begin{align} \label{eq-vb-23}
\tr_\cF\big( \rho \, P_2^*(a^*,a) \, P_2(a^*,a) \big) \ \geq \ 0 
\end{align}
of the density matrix $\rho$ on positive observables of the form
$P_2^*(a^*,a) \, P_2(a^*,a)$, where $P_2(a^*,a)$ is a polynomial of
degree 2 or smaller in the creation and annihilation operators.
A crucial difference, however, is that $\rho$ is not assumed
to be particle number-conserving, as this would not be a fair
assumption for boson systems. Hence, the reduction of the general
condition \eqref{eq-vb-23} to simpler conditions like G, P, and Q
is not as straightforward as in the fermion case and is, in fact,
not carried out in this paper, but is subject to future work.


\section{Second Quantization and Bogoliubov--Hartree--Fock Theory}\label{sec: basics}

Let $\left( \fh , \left< \cdot , \cdot \right>_{\fh} \right)$ be a complex separable Hilbert space with the inner product $\left< \cdot , \cdot \right>_{\fh}  : \, \fh \times \fh \ra \dC$. For any $N \in \NN$, the $N$-particle Hilbert space representing a physical system of $N$ indistinguishable particles is given as the $N$-fold tensor product of copies of $\fh$, \ie, 
\begin{align*}
\fh^{\otimes N} := \bigotimes_{k=1}^N \fh.
\end{align*} 
The inner product $\left< \cdot, \cdot \right>_{\fh^{\otimes N}}: \, \fh^{\otimes N} \times \fh^{\otimes N} \ra \dC$ is given by $\left< f^{(N)}, g^{(N)} \right>_{\fh^{\otimes N}} := \prod\limits_{k=1}^N \left< f_k,g_k \right>_{\fh}$ for any $f^{(N)} \equiv f_1 \otimes \cdots \otimes f_N, \ g^{(N)} \equiv g_1 \otimes \cdots \otimes g_N \in \fh^{\otimes N}$ and extension by linearity.

The Fock space $\cF$ is defined as the direct sum of all $N$-particle Hilbert spaces, 
\begin{align*}
\cF\equiv \cF[\fh] := \bigoplus_{N=0}^{\infty} \fh^{\otimes N}.
\end{align*} 
Here, by convention $\fh^{\otimes 0} := \dC$ and $\left< f^{(0)}, g^{(0)} \right>_{\fh^{(0)}} := \ol{f}^{(0)} g^{(0)}$ for $ f^{(0)}, g^{(0)} \in \fh^{(0)}$. Any vector $\Psi \in \cF$ can be written as a sequence of $N$-particle wave functions $f^{(N)} \in \fh^{\otimes N}$:
\begin{align*}
\Psi = \left( f^{(N)} \right)_{N=0}^{\infty}.
\end{align*}
The vacuum vector $\vac := \left( 1,0,0, \dots \right) \in \cF$, is considered as the basis vector of $\fh^{\otimes 0}$. With the inner product $\left< \cdot , \cdot \right>_{\cF} : \, \cF\times \cF\ra \dC$ defined by
\begin{align*}
\left< \Psi , \Phi \right>_{\cF} := \sum\limits_{N=0}^{\infty} \left< f^{(N)} , g^{(N)} \right>_{\fh^{\otimes N}}
\end{align*}
for any $\Psi \equiv \left( f^{(N)} \right)_{N=0}^{\infty}, \Phi \equiv \left( g^{(N)} \right)_{N=0}^{\infty} \in \cF$, the Fock space is a Hilbert space.

The particle number operator is defined as
\begin{align*}
\hNN := \bigoplus_{N=0}^{\infty} N \1_{\fh^{\otimes N}}.
\end{align*}
%

For any bounded operator $B \in \cB (\fh)$, $\Gamma (B) := \bigoplus_{N=0}^{\infty} B^{\otimes N}$ is an operator on $\cF$. In particular, $\Gamma(B)$ is trace class, $\Gamma (B) \in \cL^1 (\cF)$, if $B \in \cL^1 (\fh)$ and, for bosons, additionally $\sn{B}_{\cB (\fh)} \leq 1$.

A detailed description of the Fock representation can be found in \cite{BR1,BR3,BR2,So2,Th1}.


\subsection{Bosons}

The boson Fock space is the symmetric subspace of the Fock space $\cF$, \ie,
\begin{align*}
 \cF^+ \equiv \cF^+ [ \fh ] :=  S \bigoplus_{N=0}^{\infty} \fh^{\otimes N}.
\end{align*}
Here, the symmetrization operator $S \in \cB (\cF)$ is defined by
\begin{align*}
S \left( f^{(N)} \right)_{N=0}^{\infty} := \left( \frac{1}{N!} \sum\limits_{\pi \in \fS_N} \bigotimes_{k=1}^N f^{(N)}_{\pi \left( k \right)} \right)_{N=0}^{\infty}
\end{align*}
with $f^{(N)} = f^{(N)}_1 \otimes \cdots \otimes f^{(N)}_N \in \fh^{\otimes N}, \ N \in \NN$, where $\fS_N$ denotes the symmetric group with permutations $\pi$ of $N$ elements.

\begin{definition}
For any $f \in \fh$, the boson creation and annihilation operators are denoted by $\bc ( f )$ and $\ba ( f )$, respectively. Their domain, which lies dense in $\cF^+$, is $\cD (\hNN^{\frac{1}{2}}) \cap \cF^+ = \left\{ \Psi \equiv \left( f^{(N)} \right)_{N=0}^{\infty} \in \cF^+ \Big| \, \sum_{N=0}^{\infty} \left( N+1 \right) \sn{f^{(N)}}^2 < \infty \right\}$. A complete characterization of $\bc$ and $\ba$ is given by the properties
\begin{align*}
\ba ( f ) \vac = 0 \ , \quad \bc ( f ) \vac = f,
\end{align*}
and the canonical commutation relations (CCR)
\begin{subequations}
\begin{align*}
&\left[ \bc ( f  ), \bc ( g ) \right]= 0 \ , \quad \left[ \ba ( f  ), \ba ( g ) \right] = 0, \ \text{and}\\
&\left[ \ba ( f ), \bc ( g ) \right] = \left< f, g \right> \1_{\cF}
\end{align*}
\end{subequations}
for any $f,g \in \fh$, where $\1_{\cF} \in \cB(\cF)$ is the identity operator on the Fock space and $\left[ A,B \right] := AB - BA$ the commutator. 
\end{definition}

The creation operator $\bc (f)$ is linear in $f$, while the annihilation operator $\ba (f)$ is antilinear. Furthermore, the creation and annihilation operators are adjoints of each other, $\bc ( f ) = \left( \ba ( f ) \right)^*$.  Henceforth, we use the abbreviations $\bc_k \equiv \bc ( \varphi_k )$ and $\ba_k \equiv \ba ( \varphi_k )$ for a fixed, but arbitrary orthonormal basis (ONB) $\left\{ \varphi_k \right\}_{k=1}^{\infty}$ of $\fh$. 



Unlike the fermion case, the space generated by all boson creation and annihilation operators cannot be used to define a $C^*$-algebra. To this end, we introduce the Weyl operators and construct the CCR algebra. For a detailed survey, see, \eg, \cite{BR2}.

For every $f \in \fh$, we define the field operator $\Phi (f) : \, \cD (\Phi(f)) \subseteq \cF^+ \ra \cF^+$ by
\begin{align*}
\Phi (f) := \frac{1}{\sqrt{2}} \left( \bc (f) + \ba (f) \right).
\end{align*}
The field operator is essentially selfadjoint on $\cF^+$. Therefore, its closure is selfadjoint and we denote it by $\Phi (f)$, as well. 
\begin{definition}
For every $f \in \fh$, the unitary transformation $\mathds{W} (f) : \, \cF^+ \ra \cF^+$, called Weyl operator, is defined by
\begin{align*}
\mathds{W} (f) := \exp \left( \I \Phi (f) \right).
\end{align*}
\end{definition}

The Weyl operators satisfy $\mathds{W}(f)^* = \mathds{W} (-f)$ and the Weyl commutation relations
\begin{align*}
\mathds{W} (f) \: \mathds{W} (g) = \e^{- \frac{\I}{2} \rIm \left< f , g \right>_{\fh}} \mathds{W} (f+g)
\end{align*}
for any $f,g \in \fh$. The commutator of two Weyl operators is completely determined by the Weyl commutation relations. Furthermore, we have $\mathds{W} (0) = \1_{\cF^+}$.

A given field operator $\Phi (f)$ with $f \in \fh$ is transformed by the Weyl operator $\mathds{W} (g), \ g \in \fh$, as
\begin{align*}
\mathds{W} (g) \: \Phi (f) \: \mathds{W}(g)^* = \Phi (f) - \rIm \left< g,f \right>_{\fh} \1_{\cF}.
\end{align*}
Hence, $\mathds{W}_g \equiv \mathds{W} ( \I \sqrt{2} g )$ defines a unitary transformation, called Weyl transformation, for any $g \in \fh$. For any $f \in \fh$, this transformation yields
\begin{align*}
\mathds{W}_g \: \bc (f) \: \mathds{W}_g^* = \bc (f) + \left< g,f \right>_{\fh} \qquad \text{and} \qquad \mathds{W}_g \: \ba (f) \: \mathds{W}_g^* = \ba (f) + \left< f,g \right>_{\fh}.
\end{align*}

\begin{definition}
The $C^*$-algebra $\cW$ generated by $\left\{ \mathds{W} (f) \Big| \, f \in \fh \right\}$ is called Weyl algebra or CCR algebra.
\end{definition} 
This algebra is unique up to $*$-automorphisms (Cf. Theorem~5.2.8. of \cite{BR2}). 


\subsubsection*{Boson Bogoliubov Transformation}

\begin{remark}\label{rem: def Complex Conjugates}
The following definition of the Bogoliubov transformation depends on the choice of the ONB $\left\{ \varphi_k \right\}_{k=1}^{\infty}$ of $\fh$, since the complex conjugate of a function $f \in \fh$, that is given by $f = \sum\limits_{k=1}^{\infty} \mu_k \varphi_k$ with some $\mu_k \in \dC, k \in \NN$, is defined by 
\begin{align}
\overline{f} := \sum\limits_{k=1}^{\infty} \overline{\mu}_k \varphi_k.\label{eq: bos complex conjugate vector}
\end{align} 
Furthermore, we define for any operator $A$ the complex conjugate operator $\overline{A}$ by
\begin{align}
\left< f, \overline{A} \, g \right> := \overline{\left< \ol{f}, A \, \ol{g} \right>}.\label{eq: bos complex conjugate operator}
\end{align}
\end{remark}

We emphasize that there is also a formulation to obtain a basis independent definition of the Bogoliubov transformation, \eg, in \cite{So2,Na1}. There, the underlying space is $\fh \oplus \fh^*$ instead of $\fh \oplus \fh$ and an antilinear map $J : \, \fh \ra \fh^*$, defined by $J g (f) := \left< g , f \right>_{\fh}$ for any $f,g \in \fh$, and its inverse $J^* : \, \fh^* \ra \fh$ are required. Then, the second component of a vector $f \oplus g \in \fh \oplus \fh$ is replaced by $Jg \in \fh^*$ such that $f \oplus Jg \in \fh \oplus \fh^*$. The new vector is antilinear in the second component. This supersedes the definition of the complex conjugate of a function. Furthermore, some operators map from $\fh^*$ to $\fh$ or vice versa, \eg, $v : \, \fh \ra \fh^*$. Then, for instance, $\ol{v}$ and $\ol{u}$ of the following definition are replaced by the maps $v : \, \fh \ra \fh^*$ and $J \, u \, J^* : \, \fh^* \ra \fh^*$, respectively.

For any linear operator $A$, the transpose is defined as $A^T := \ol{A}^* = \ol{A^*}$. Now, we are prepared to define a boson Bogoliubov transformation.

\begin{definition}\label{def: bosBogoliubov}
A linear map $U = \left( \begin{smallmatrix} u & v \\ \ol{v} & \ol{u} \end{smallmatrix} \right) : \, \fh \oplus \fh \ra \fh \oplus \fh$ is called boson Bogoliubov transformation if the linear operators $u : \, \fh \ra \fh$ and $v : \, \fh \ra \fh$ fulfill
\begin{subequations}\label{eq: def bosBogoliubov}
\begin{align}
u u^* - v v^* = \1_{\fh}, \quad u^* u - v^T \ol{v} = \1_{\fh},\label{eq: def bosBogoliubov 1}\\
u^* v - v^T \ol{u} = 0, \quad u v^T - v u^T = 0.\label{eq: def bosBogoliubov 2}
\end{align}
\end{subequations}
\end{definition}

\begin{remark}
Any boson Bogoliubov transformation $U$ is invertible. The inverse is given by the Bogoliubov transformation
\begin{align*}
U^{-1} = \cS \, U^* \cS
\end{align*} 
with
\begin{align*}
\cS := \begin{pmatrix} \1_{\fh} & 0 \\ 0 & -\1_{\fh} \end{pmatrix}. 
\end{align*}
\end{remark}

\begin{remark}
Eqs.~(\ref{eq: def bosBogoliubov}a,b) on $u$ and $v$ are equivalent to stating
\begin{align*}
U^* \cS \, U = \cS, \quad
U \, \cS \, U^* = \cS.
\end{align*}
\end{remark}

\begin{lemma}\label{lem: bosBogUnitaryImplementation}
Let $U = \left( \begin{smallmatrix} u & v \\ \ol{v} & \ol{u} \end{smallmatrix} \right) : \, \fh \oplus \fh \ra \fh \oplus \fh$ be a boson Bogoliubov transformation. There is a unitary transformation $\mathds{U}_{U}: \cF^+ \ra \cF^+$ such that
\begin{align*}
\mathds{U}_{U} \left[ \bc ( f ) + \ba ( \ol{g} ) \right] \mathds{U}_{U}^* = \bc ( u f + v g ) + \ba ( v \ol{f} + u \ol{g} )
\end{align*}
for all $f,g \in \fh$ if and only if $v$ is Hilbert--Schmidt. We call $\mathds{U}_{U}$ unitary representation or implementation of $U$ on $\cF^+$.
\end{lemma}

The condition that $v$ is Hilbert--Schmidt is named after Shale and Stinespring \cite{SS1}.


\subsubsection*{States and Density Matrices}

Next, we introduce the notion of states and, afterwards, of density matrices.

\begin{definition}
A continuous linear functional $\omega \in \cW^*$ on the CCR algebra $\cW$ is called a state if it is normalized and positive, \ie, $\omega (\1_{\cF}) = 1$ and $\omega (A) \geq 0$ for all positive semi-definite operators $A \in \cW$. 
\end{definition}

Since the boson creation and annihilation operators are not in $\cW$, their expectation values are not well-defined for all states. In order to find well-defined expressions for these expectation values, first, we restrict ourselves to specific states and, then, extend the domain for these states appropriately.

Let $\mathds{W} (f)$ denote a Weyl operator for any $f \in \fh$ and let $\omega$ be a state. We assume that the map $T_f : \, \dR \ra \dC , \ t \mapsto \omega \big( \mathds{W} (tf) \big)$ is four times continuously differentiable for all $f \in \fh$, shortly $T_f \in \cC^4 (\dR;\dC)$. This assumption provides the definition of the expectation value of a single creation or annihilation operator and of the particle number operator. E.g., we have
\begin{align*}
\omega \big( \Phi (f) \big) := \frac{\di}{\di t} \, \omega \big( \mathds{W} (tf) \big) \Big|_{t=0} < \infty 
\end{align*}
for any $f \in \fh$ and, hence, by linearity of $\omega$
\begin{align*}
\omega ( \ba (f) ) = \frac{1}{\sqrt{2}} \left[ \omega \big( \Phi (f) \big) + \I \omega \big( \Phi (\I f) \big) \right].
\end{align*}
Analogously, we give a meaning to
\begin{align*}
\omega \big( \bc (f) \big), \quad \omega \big( e (f) \, e (g) \big),\quad \text{and} \quad \omega \big( e (f_1) \, e (f_2) \, e (g_2) \, e (g_1) \big)
\end{align*}
for $f,g,f_1,f_2,g_1,g_2 \in \fh$, due to $T_f \in \cC^4 (\dR;\dC)$. Here, $e$ denotes either the creation operator $\bc$ or the annihilation operator $\ba$. Thus, the expectation value in this state $\omega$ can be defined not only for elements of the CCR-algebra, but also for polynomials of degree 4 in creation and annihilation operators. In order to exemplify such polynomials, we note that, for instance, a general polynomial of degree 2 can be written as 
\begin{align*}
\Po_2 := \sum\limits_{k,l=1}^N \left[ \alpha_{kl} \bc_k \ba_l + \beta_{kl} \bc_k \bc_l + \epsilon_{kl} \ba_k \ba_l \right] + \sum\limits_{k=1}^N \left[ \zeta_k \bc_k + \xi_k \ba_k \right] + \mu
\end{align*} 
with an ONB $\left\{ \varphi_k \right\}_{k=1}^{\infty}$ of $\fh$, some $N \in \NN$, and coefficients $\alpha_{kl}, \beta_{kl}, \epsilon_{kl}, \zeta_k, \xi_k, \mu \in \dC$. 

\begin{definition}
We denote the closure of the indicated extension to the polynomials of degree 4 in creation and annihilation operators by $\cA^+$.
\end{definition} 

For any ONB $\left\{ \varphi_k \right\}_{k=1}^{\infty}$ of $\fh$, the monotonously increasing sequence of the polynomials $\hNN_N := \sum\limits_{k=1}^N \bc_k \ba_k$, $N \in \NN$, converges strongly to the particle number operator $\hNN$ on $\cD (\hNN) \cap \cF^+ = \left\{ \Psi \equiv \left( f^{(N)} \right)_{N=0}^{\infty} \in \cF^+ \Big| \, \sum\limits_{N=0}^{\infty} \left( N+1 \right)^2 \sn{f^{(N)}}^2 < \infty \right\}$.

\begin{definition}
Let $\omega$ be a state. If $T_f \in \cC^4 ( \dR;\dC )$ for any $f \in \fh$ and $\omega \big( \hNN^2 \big) := \lim\limits_{N \ra \infty} \omega \big( \hNN_N^2 \big) < \infty$, we write $\omega \in \cZ^+$. 
\end{definition}

\begin{remark}
For any $\omega \in \cZ^+$, the Cauchy--Schwarz inequality yields
\begin{align*}
\omega \big( \hNN \big) \leq \sqrt{\omega \big( \hNN^2 \big)} < \infty.
\end{align*}
\end{remark}

\begin{definition}
A state $\omega \in \cZ^+$ is called pure if there is a $\Psi \in \cF^+$, such that, for any $A \in \cA^+$, 
\begin{align*}
\omega ( A ) = \left< \Psi, A \Psi \right>_{\cF}.
\end{align*}
\end{definition} 

\begin{definition}\label{def: bos centered state}
A centered state is a state $\omega \in \cZ^+$ with 
\begin{align}
\omega \big( \bc (f) \big) = 0\label{eq: def bos centered state}
\end{align}
for any $f \in \fh$. We denote the set of all centered states by $\cZ_{\mathrm{cen}}^+$.
\end{definition}
As follows from \eqref{eq: def bos centered state}, we also have $\omega \big( \ba (f) \big) = 0$ for $\omega \in \cZ_{\mathrm{cen}}^+$. 

\begin{definition}\label{def: bos quasifree}
A state $\omega \in \cZ^+$ is called quasifree, shortly $\omega \in \cZ_{\mathrm{qf}}^+$, if there is a positive semi-definite operator $h_{\omega}$ on $\fh$ and $f_{\omega} \in \fh$ such that, for every $f \in \fh$,
\begin{align*}
\omega \big( \mathds{W}_f \big) = \exp \left( - 2 \I \left< f_{\omega} , f \right>_{\fh} - \left< f, \left( \1_{\fh} + h_{\omega} \right) f \right>_{\fh} \right).
\end{align*}
The subset of pure quasifree states is denoted by $\cZ_{\mathrm{pqf}}^+$.
\end{definition}

\begin{remark}
$\cZ_{\mathrm{qf}}^+$ and $\cZ_{\mathrm{pqf}}^+$ are invariant under Bogoliubov and Weyl transformations, \ie, the transform of a (pure) quasifree state is (pure) quasifree, as well.
\end{remark}


\begin{definition}
We say that a state $\omega \in \cZ^+$ is coherent, $\omega \in \cZ_{\mathrm{coh}}^+$, if there is an $f \in \fh$ such that for all $A \in \cA^+$
\begin{align*}
\omega (A) = \left< \vac, \mathds{W}_f  A \, \mathds{W}_f^* \, \vac \right>_{\cF}.
\end{align*}
\end{definition} 

\begin{remark}
In particular, we have
\begin{align*}
\cZ_{\mathrm{cqf}}^+ \subsetneq \cZ_{\mathrm{qf}}^+\ \text{and} \
\cZ_{\mathrm{coh}}^+ \subsetneq \cZ_{\mathrm{pqf}}^+ \subsetneq \cZ_{\mathrm{qf}}^+.
\end{align*}
\end{remark}

\begin{definition}
We call a positive semi-definite operator $\rho \in \cL^1 ( \cF^+)$ with $\tr_{\cF^+} \left( \rho \right) = 1$ a density matrix. 
\end{definition}
For a density matrix $\rho \in \cL^1 ( \cF^+)$, the map $ \cA^+ \ra \dC, A \mapsto \tr_{\cF^+} \left( \rho^{\frac{1}{2}} A \rho^{\frac{1}{2}} \right)$ defines a state. In particular, for every state $\omega \in \cZ^+$, there is a density matrix $\rho$ with $\tr_{\cF^+} \left( \rho^{\frac{1}{2}} A \rho^{\frac{1}{2}} \right) = \omega (A)$ for all $A \in \cA^+$. Therefore, the notions of pureness, quasifreeness etc can be transferred to the corresponding density matrix.


\subsubsection*{One- and Two-Particle Density Matrices and Representability}\label{sec: Boson-pdm-repr}

For systems with pair-interactions, the formulation of the variational problem can be reduced by the notion of one- and two-particle density matrices.

\begin{definition}
For any state $\omega \in \cZ^+$, the corresponding (boson) one-particle density matrix (1-pdm) $\gamma_{\omega} : \, \fh \ra \fh$ is defined by its matrix elements
\begin{align}\label{eq: bos 1pdm}
\left< f, \gamma_{\omega} \, g \right>_{\fh} := \omega \big( \bc (g) \, \ba (f) \big)
\end{align}
for every $f,g \in \fh$. 
\end{definition}
Since any state is positive, we have
\begin{align*}
\left< f, \gamma_{\omega} \, f \right>_{\fh} = \omega \big( \bc (f) \, \ba (f) \big) \geq 0
\end{align*}
for any $f \in \fh$. Hence, the 1-pdm is a selfadjoint and positive semi-definite operator. Moreover, $\gamma_{\omega} \in \cL^1 (\fh)$ due to
\begin{align*}
\tr_{\fh} \left( \gamma_{\omega} \right) &= \sum\limits_{k=1}^{\infty} \left< \varphi_k, \gamma_{\omega} \, \varphi_k \right>_{\fh} = \omega \Big( \sum\limits_{k=1}^{\infty} \bc_k \ba_k \Big) = \omega \big( \hNN \big) < \infty.
\end{align*} 

\begin{definition}
The (boson) two-particle density matrix (2-pdm) $\Gamma_{\omega} : \, \fh \otimes \fh \ra \fh \otimes \fh$ of a state $\omega \in \cZ^+$ is defined by
\begin{align*}
\left< f_1 \otimes f_2, \Gamma_{\omega} \left( g_1 \otimes g_2 \right) \right> := \omega \big( \bc ( g_2 ) \, \bc ( g_1 ) \, \ba ( f_1 ) \, \ba ( f_2 ) \big)
\end{align*}
for any $f_1,f_2,g_1,g_2 \in \fh$. 
\end{definition}
The 2-pdm is a selfadjoint and positive semi-definite trace class operator since
\begin{align*}
\tr_{\fh \otimes \fh} \left( \Gamma_{\omega} \right) = \omega \big( \hNN^2 - \hNN \big) < \infty
\end{align*}
and, for any $\psi \equiv \sum_{k,l=1}^{\infty} \mu_{kl} \varphi_k \otimes \varphi_l \in \fh \otimes \fh$, $\mu_{kl} \in \dC$, 
\begin{align*}
\left< \psi , \Gamma_{\omega} \, \psi \right>_{\fh \otimes \fh} = \sum\limits_{i,j,k,l=1}^{\infty} \mu_{kl} \ol{\mu}_{ij} \omega \big( \bc_k \bc_l \ba_j \ba_i \big) = \omega \big(P P^*\big) \geq 0, 
\end{align*}
where $P := \sum_{k,l=1}^{\infty} \mu_{kl} \bc_k \bc_l$.
Furthermore, the 2-pdm is symmetric, \ie, $\Gamma_{\omega} \, \Ex = \Ex \, \Gamma_{\omega} = \Gamma_{\omega}$ for all $f,g \in \fh$. Here, the exchange operator $\Ex : \, \fh \otimes \fh \ra \fh \otimes \fh$ is the linear map defined by
\begin{align*}
\Ex \left( f \otimes g \right) := g \otimes f 
\end{align*}
for any $f,g \in \fh$. Summarizing the basic properties of the 1- and 2-pdm, we introduce the notions of admissibility and representability.

\begin{definition}
We call a pair $\left( \gamma, \Gamma \right)$ of operators on $\fh \times \left( \fh \otimes \fh \right)$ admissible if 
\begin{itemize}
\item[(i)] $\Gamma \in \cL^1 \left( \fh \otimes \fh \right)$ is symmetric, i.e., $\Ex \, \Gamma = \Gamma \, \Ex = \Gamma$, and selfadjoint, and
\item[(ii)] $\gamma \in \cL^1 \left( \fh \right)$ with $\tr_{\fh} \left( \gamma \right) = \omega \big( \hNN \big)$ is selfadjoint and positive semi-definite.
\end{itemize}
\end{definition}

\begin{definition}
We say that the pair $\left( \gamma, \Gamma \right)$ of operators on $\fh \times \left( \fh \otimes \fh \right)$ is representable if there is a state $\omega \in \cZ^+$ with $\gamma_{\omega} = \gamma$ and $\Gamma_{\omega} = \Gamma$.\\
Necessary conditions on the pair $\left( \gamma, \Gamma \right)$ to be representable are called re\-pre\-sen\-ta\-bi\-li\-ty conditions.
\end{definition}

In particular, every representable pair $\left( \gamma, \Gamma \right)$ is admissible. 


\subsubsection*{Generalized One- and Two-Particle Density Matrices for Bosons}\label{sec: Boson-pdm}

In \cite{BLS}, a generalized 1-pdm is defined for fermions on the space $\fh \oplus \fh$. We provide a definition of the generalized 1-pdm for bosons and, then, further generalize the one- and the two-particle density matrices. Here, again, the definitions depend on the choice of the ONB of $\fh$, since we use complex conjugates of functions, as well as operators as explained in Remark~\ref{rem: def Complex Conjugates}. We refer the reader to \cite{So2} for a basis independent formulation.

\begin{definition}\label{def: bos gen 1-pdm}
For any state $\omega \in \cZ^+$, the generalized 1-pdm $\tgamma_{\omega}$ is an operator on $\fh \oplus \fh$ defined by 
\begin{align}\label{eq: bos def gen 1-pdm}
\left< \left( f_1 \oplus f_2 \right), \tgamma_{\omega} \left( g_1 \oplus g_2 \right) \right> := \omega \big( \left[ \bc ( g_1 ) + \ba ( \ol{g}_2 ) \right] \left[ \ba ( f_1 ) + \bc ( \ol{f}_2 ) \right] \big)
\end{align}
for $f_1, f_2, g_1, g_2 \in \fh$.
\end{definition} 

\begin{remark}\label{rem: bos gen 1pdm matrix}
Defining
\begin{align}\label{eq: bos alpha}
&\alpha_{\omega}^* : \, \fh \ra \fh \ , \ \left< f, \alpha_{\omega}^* \, g \right> := \omega \big( \bc ( g ) \, \bc ( \ol{f} ) \big),
\end{align}
we are able to write the generalized 1-pdm as 
\begin{align*}
\tgamma_{\omega} = \begin{pmatrix} \gamma_{\omega} & \alpha_{\omega} \\ \alpha_{\omega}^* & \1_{\fh}+\overline{\gamma}_{\omega} \end{pmatrix},
\end{align*}
a matrix with operator-valued entries. The 1-pdm $\gamma_{\omega}$ is selfadjoint and $\alpha_{\omega}$ symmetric, \ie, $\alpha_{\omega}^T = \alpha_{\omega}$ for its transpose $\alpha_{\omega}^T := \ol{\alpha_{\omega}^*}$.
\end{remark}

\begin{lemma}
For any $\omega \in \cZ^+$, the generalized 1-pdm $\tgamma_{\omega}$ as defined in \eqref{eq: bos def gen 1-pdm} is a positive semi-definite operator on $\fh \oplus \fh$. In particular, it is selfadjoint.
\end{lemma}

\begin{proof}
By setting $g_1 = f_1$ and $g_2 = f_2$, the assertion is a consequence of \eqref{eq: bos def gen 1-pdm} and the positivity of the corresponding state.
\end{proof}

Therefore, the boson 1-pdm $\gamma$ is positive semi-definite, too. Unlike the fermion case, the boson 1-pdm is not bounded above by $\1_{\fh}$.

So far, the definitions and statements are well established and can be found for example in \cite{Na1}, in \cite{So2} for both particle types and, in a version for fermions, in \cite{BLS}. 

\begin{lemma}\label{lem: bos transformed 1pdm}
Let $\omega \in \cZ^+$ be a state and $\,\tgamma : \, \fh \oplus \fh \ra \fh \oplus \fh$ its generalized 1-pdm. For a boson Bogoliubov transformation $U : \, \fh \oplus \fh \ra \fh \oplus \fh$ with unitary representation $\mathds{U}_{U} : \, \cF^+ \ra \cF^+$, define $\omega_{U}$ by $\omega_{U} ( A ) := \omega \big( \mathds{U}_{U} A \, \mathds{U}_{U}^* \big)$ for any $A \in \cA^+$. Then, the generalized 1-pdm $\tgamma_{U}$ of the state $\omega_{U}$ is given by
\begin{align}
\tgamma_{U} = U^* \, \tgamma \, U.\label{eq: bos transformed 1pdm}
\end{align}
Furthermore, $\tgamma \, \cS \, \tgamma = - \tgamma$ implies $\tgamma_{U} \, \cS \, \tgamma_{U} = - \tgamma_{U}$.
\end{lemma}

\begin{proof}
We consider the matrix elements of $\tgamma_{U}$. For the first assertion, we obtain 
\begin{align*}
\left< \left( f_1 \oplus f_2 \right), \tgamma_{U} \left( g_1 \oplus g_2 \right) \right> &= \omega \left( \mathds{U}_{U} \left[ \bc ( g_1 ) + \ba ( \ol{g}_2 ) \right] \mathds{U}_{U}^* \mathds{U}_{U} \left[ \ba ( f_1 ) + \bc ( \ol{f}_2 ) \right] \mathds{U}_{U}^* \right) \notag \\
&= \omega \left( \left[ \bc ( u g_1 + v g_2 ) + \ba ( u \ol{g}_2 + v \ol{g}_1 ) \right] \right. \notag \\
&\relphantom{= \omega....................} \times \left. \left[ \ba ( u f_1 + v f_2 ) + \bc ( u \ol{f}_2 + v \ol{f}_1 ) \right] \right) \notag \\
& = \left< U \left( f_1 \oplus f_2 \right), \tgamma \, U \left( g_1 \oplus g_2 \right) \right> 
\end{align*}
for any $f_1, f_2, g_1, g_2 \in \fh$. Thus, \eqref{eq: bos transformed 1pdm} holds. 

The second assertion follows from \eqref{eq: bos transformed 1pdm} and $U \cS U^* = \cS$:
\begin{align*}
\tgamma_{U} \, \cS \, \tgamma_{U} = U^* \, \tgamma \, U  \cS \, U^* \, \tgamma \, U = U^* \, \tgamma \, \cS \, \tgamma \, U = - U^* \, \tgamma \, U = - \tgamma_{U},
\end{align*}
which completes the proof.
\end{proof}

In the following, we give a further generalization of the 1-pdm on the space $\fH_{\mathrm{gen}} := \fh \oplus \fh \oplus \dC$. 

\begin{definition}\label{def: bos further gen 1-pdm}
For any state $\omega \in \cZ^+$, the further generalized 1-pdm $\hgamma_{\omega} : \, \fH_{\mathrm{gen}} \ra \fH_{\mathrm{gen}}$ is defined by
\begin{align}\label{eq: def bos further gen 1-pdm}
\left< G, \hgamma_{\omega} F \right> := \omega \big( \big[ \bc ( f_1 ) + \ba ( \ol{f}_2 ) + \mu \big] \big[ \ba ( g_1 ) + \bc ( \ol{g}_2 ) + \ol{\nu} \big] \big)
\end{align}
for $F \equiv f_1 \oplus f_2 \oplus \mu$ and $G \equiv g_1 \oplus g_2 \oplus \nu \in \fH_{\mathrm{gen}}$. 
\end{definition}

\begin{remark}\label{rem: bos further gen 1pdm matrix}
We rewrite the further generalized 1-pdm $\hgamma_{\omega}$ as a $3 \times 3$-matrix:
\begin{align}
\hgamma_{\omega} = \begin{pmatrix} \gamma_{\omega} & \alpha_{\omega} & b_{\omega} \\ \alpha_{\omega}^* & \1_{\fh} + \ol{\gamma}_{\omega} & \ol{b}_{\omega} \\ b_{\omega}^* & \ol{b}_{\omega}^* & 1 \end{pmatrix}.\label{eq: bos further gen 1pdm matrix}
\end{align}
Here, the first moment $b_{\omega} \in \fh$ and its dual element $b_{\omega}^* \in \fh^*$ are given by
\begin{align}\label{eq: bos first moment}
\left< g , b_{\omega} \right>_{\fh} := \omega \big( \ba (g) \big) \qquad \text{and} \qquad b_{\omega}^* \cdot g \equiv \left< b_{\omega} , g \right>_{\fh} = \omega \big( \bc (g) \big)
\end{align}
for every $g \in \fh$.
For the complex conjugate $\ol{b}_{\omega}$ of the wave function $b_{\omega} \in \fh$, we have $\left< g , \ol{b}_{\omega} \right>_{\fh} = \ol{\left< \ol{g} , b_{\omega} \right>}_{\fh} = \ol{\omega \big( \ba (\ol{g}) \big)} = \omega \big( \bc (\ol{g}) \big)$.
\end{remark}

\begin{proposition}
The further generalized 1-pdm is positive semi-definite and selfadjoint.
\end{proposition}

\begin{proof}
The selfadjointness is a direct consequence of \eqref{eq: bos further gen 1pdm matrix}. By setting $F = G$ in \eqref{eq: def bos further gen 1-pdm}, $\hgamma \geq 0$ follows from the positivity of the state $\omega$.
\end{proof}

\begin{lemma} 
Let $\hgamma = \left( \begin{smallmatrix} \gamma & \alpha & b \\ \alpha^* & \1_{\fh} + \ol{\gamma} & \ol{b} \\ b^* & \ol{b}^* & 1 \end{smallmatrix} \right)  : \, \fh \oplus \fh \oplus \dC \ra \fh \oplus \fh \oplus \dC$ be a positive semi-definite trace class operator with $b \in \fh, \ \gamma \in \cL^1 (\fh)$, and $\alpha \in \cL^2 (\fh)$. Then, there is a unique quasifree state $\omega$ that has $\hgamma$ as its further generalized 1-pdm.\\
In particular, for any positive semi-definite trace class operator $\tgamma = \left( \begin{smallmatrix} \gamma & \alpha \\ \alpha^* & \1_{\fh} + \ol{\gamma} \end{smallmatrix} \right) : \, \fh \oplus \fh \ra \fh \oplus \fh$ with $\gamma \in \cL^1 (\fh)$ and $\alpha \in \cL^2 (\fh)$, there is an $\omega \in \cZ_{\mathrm{cqf}}^+$ with $\tgamma = \tgamma_{\omega}$.
\end{lemma}

\begin{proof}
The second part is a consequence of Theorem~11.4 in \cite{So2} or Theorem~1.6 (i) in \cite{Na1}. The first part follows from the second part due to the fact that a non-centered state with first moment $b \in \fh$ is completely characterized by a Weyl operator $\mathds{W}_b$ and the centered state $\omega_0$ defined by $\omega_0 (A) := \omega (\mathds{W}_b^* A \mathds{W}_b)$ for any $A \in \cA^-$. 
\end{proof}


Analogously, we define a generalized 2-particle density matrix $\hGamma$ on 
\begin{align*}
\fH_{\mathrm{sim}} := \left( \bigoplus\limits_{n=1}^4 \fh \otimes \fh \right) \oplus \left( \bigoplus\limits_{n=1}^2 \fh \right) \oplus \dC.
\end{align*}
Technically, the generalized 2-pdm should be defined on
\begin{align*}
\fH_{\mathrm{gen}} \otimes \fH_{\mathrm{gen}} &\cong \left( \bigoplus\limits_{n=1}^4 \fh \otimes \fh \right) \oplus \left( \bigoplus\limits_{n=1}^4 \fh \right) \oplus \dC. 
\end{align*}
It suffices, however, to consider $\fH_{\mathrm{sim}}$, since for any polynomial of degree 1 in annihilation and creation operators there are an ONB $\left\{ \varphi_k \right\}_{k=1}^{\infty}$ of $\fh$, an $N \in \NN$, and coefficients $\mu_k, \nu_k, \sigma_k \tau_k, \tilde{\mu}_k, \tilde{\nu}_k \in \dC, k = 1,\dots,N,$ such that 
\begin{align*}
\sum\limits_{k=1}^N \left( \mu_k \bc_k + \nu_k \ba_k + \overline{\sigma}_k \ba_k + \overline{\tau}_k \bc_k \right)  = \sum\limits_{k=1}^N \left( \tilde{\mu}_k \bc_k + \tilde{\nu}_k \ba_k \right).
\end{align*} 

For any $M \in \NN$ and a given ONB $\left\{ \varphi_k \right\}_{k=1}^{\infty}$ of $\fh$, we set $F:=\left(F_1,F_2,F_3,F_4\right)^T,\ G:=\left(G_1,G_2,G_3,G_4\right)^T,\ f:=\left(f_1,f_2\right)^T,$ and $g:=\left(g_1,g_2\right)^T$ with $F_i := \sum_{k,l = 1}^{M} \mu_{kl}^{(i)} \varphi_k \otimes \varphi_l, \ G_i := \sum_{k,l = 1}^{M} \nu_{kl}^{(i)} \varphi_k \otimes \varphi_l \in \fh \otimes \fh,\ f_j := \sum_{k = 1}^{M} \mu_{k}^{(j)} \varphi_k$, and $g_j := \sum_{k = 1}^{M} \nu_{k}^{(j)} \varphi_k \in \fh$, where the coefficients $\mu_{kl}^{(i)},\ \nu_{kl}^{(i)},\ \mu_{k}^{(j)},\ \nu_{k}^{(j)} \in \dC$ with $k,l \in \left\{ 1,\dots,M \right\},\ i \in \left\{ 1,2,3,4 \right\},\ j \in \left\{ 1,2 \right\}$. Then, we define the polynomials $\Po_1$ and $\Po_2$ by
\begin{align*}
\Po_1 \left( f \right) &:= \sum\limits_{k=1}^{M} \left( \mu_{k}^{(1)} \bc_k + \overline{\mu}_{k}^{(2)} \ba_k \right), \\
\Po_2 \left( F \right) &:= \sum\limits_{k,l=1}^{M} \left( \mu_{kl}^{(1)} \bc_k \bc_l + \mu_{kl}^{(2)} \bc_k \ba_l + \mu_{kl}^{(3)} \ba_k \bc_l + \mu_{kl}^{(4)} \ba_k \ba_l \right).
\end{align*}

\begin{definition}\label{def: bos gen 2pdm}
The generalized 2-pdm $\hGamma_{\omega}$ is defined by
\begin{align*}
\left< \begin{pmatrix} G \\ g \\ \mu \end{pmatrix},\  \hGamma_{\omega} \begin{pmatrix} F \\ f \\ \nu \end{pmatrix} \right>  := \omega \big( \left( \Po_2 \left( F \right) + \Po_1 \left( f \right) + \nu \right) \left( \Po_2^* \left( G \right) + \Po_1^* \left( g \right) + \overline{\mu} \right) \big)
\end{align*}
for any $F,G \in \bigoplus^4 \left( \fh \otimes \fh \right)$ with $\sum_{k=1}^{\infty} \mu_{kk}^{(3)},\sum_{k=1}^{\infty} \nu_{kk}^{(3)} < \infty,\ f,g \in \fh \oplus \fh$, and $\mu,\nu \in \dC$ as an operator on $\fH_{\mathrm{sim}}$. The polynomials $\Po_1$ and $\Po_2$ are of the form specified above. 
\end{definition}

As for the generalized 1-pdm, an easy consequence of the definition are the following properties.

\begin{proposition}\label{prop: gen 2-pdm is selfadjoint and nonnegative}
The generalized 2-pdm is selfadjoint and positive semi-definite.
\end{proposition}

An explicit form of the generalized 2-pdm as a $7 \times 7$-matrix is given in Appendix~\ref{sec: appendix}. 


\subsection{Fermions}

The fermion Fock space $\cF^- \equiv \cF^- [\fh]$ is defined to be the orthogonal sum
\begin{align*}
 \cF^-[\fh]:=\bigoplus_{N=0}^{\infty} \fh^{\wedge N},
\end{align*}
where, for $N \in \NN$,
\begin{align*}
\fh^{\wedge N} := A_N \fh^{\otimes N}
\end{align*}
is the antisymmetric tensor product of $N$ copies of $\fh$ and $\fh^{\wedge 0} := \dC$. Here, the antisymmetrization operator $A \in \cB (\cF), \ A := \bigoplus_{N=0}^{\infty} A_N$ with $A_{N} : \, \fh^{\otimes N} \ra \fh^{\otimes N}$, is uniquely defined by
\begin{align*}
A_N \left( f_1 \otimes \cdots \otimes f_N \right) := \frac{1}{N!} \sum\limits_{\pi \in \fS_N} \left( -1 \right)^{\pi} f_1 \otimes \cdots \otimes f_N =: \frac{1}{\sqrt{N!}}f_1 \wedge \cdots \wedge f_N,
\end{align*}
for $f_1,\dots,f_N \in \fh$, where $\left( -1 \right)^{\pi}$ denotes the sign of the permutation $\pi \in \fS_N$. 

\begin{definition}
For any $f \in \fh$, the fermion creation and annihilation operators are denoted by $\fc ( f )$ and $\fa ( f )$, respectively. They are bounded operators on $\cF^-$. Introducing the anti-commutator $\left\{ A,B \right\} := AB + BA$, they are completely characterized by the properties
\begin{align*}
\fa ( f ) \vac = 0 \ , \quad \fc ( f ) \vac = f,
\end{align*}
and the canonical anti-commutation relations (CAR)
\begin{align*}
&\left\{ \fc ( f  ), \fc ( g ) \right\}= 0, \quad \left\{ \fa ( f  ), \fa ( g ) \right\} = 0, \ \text{and} \\
&\left\{ \fa ( f ), \fc ( g ) \right\} = \left< f, g \right>_{\fh} \1_{\cF}
\end{align*}
for any $ f,g  \in \fh$.  
\end{definition}

\begin{definition}
The $C^*$-algebra $\cA^-$ generated by $\left\{ 1, \fc (f), \fa (f) \big| \, f \in \fh \right\}$ is called CAR algebra.
\end{definition}

Let $\left\{ \varphi_k \right\}_{k=1}^{\infty}$ be a given ONB of $\fh$ and, for this basis, $\fc_k \equiv \fc(\varphi_k)$ and $\fa_k \equiv \fa(\varphi_k)$. For any $N\in\NN$, an ONB of $N$-particle Hilbert space $\fh^{\wedge N}$ is given by
\begin{align*}
\left\{  \fc_{k_1} \cdots \fc_{k_N} \vac \Big| \, 1 \leq k_1 < \cdots < k_N \right\}. 
\end{align*}
Moreover,
\begin{align*}
 \left\{  \fc_{k_1} \cdots \fc_{k_N} \vac \Big| \, N \in \NN \cup \{0\},\ 1\leq k_1< \dots < k_N \right\}
\end{align*}
is an ONB of the fermion Fock space $\cF^-$. 


\subsubsection*{Fermion Bogoliubov Transformation}

In this section, we fix an (arbitrary) orthonormal basis $\left\{ \varphi_k \right\}_{k=1}^{\infty}$ of $\fh$. As the definitions of complex conjugates of both a wave function and an operator depend on the choice of the ONB of $\fh$, so do the transforms defined in the following. We refer the reader to \cite{So2} for a basis independent formulation.

\begin{definition}\label{def: ferm Bogoliubov}
A linear map $U = \left( \begin{smallmatrix} u & v \\ \ol{v} & \ol{u} \end{smallmatrix} \right) : \, \fh \oplus \fh \ra \fh \oplus \fh$ is called fermion Bogoliubov transformation if $u : \, \fh \ra \fh$ and $v : \, \fh \ra \fh$ are two linear maps fulfilling 
\begin{subequations}\label{eq: def ferm Bogoliubov}
\begin{align}
u u^* + v v^* = \1_{\fh}, \quad u^* u + v^T \ol{v} = \1_{\fh},\\
u^* v + v^T \ol{u} = 0, \quad u v^T + v u^T = 0.
\end{align}
\end{subequations}
\end{definition}

\begin{remark}
Eqs.~(\ref{eq: def ferm Bogoliubov}a,b) on $u$ and $v$ are equivalent to the condition that $U$ is unitary, \ie,
\begin{align*}
U^*\, U = \1_{\fh \oplus \fh}, \quad U \, U^* = \1_{\fh \oplus \fh}.
\end{align*}
Therefore, the inverse of a fermion Bogoliubov transformation $U$ exists and is the fermion Bogoliubov transformation $U^*$.
\end{remark}

\begin{lemma}
Let $U = \left( \begin{smallmatrix} u & v \\ \ol{v} & \ol{u} \end{smallmatrix} \right) : \, \fh \oplus \fh \ra \fh \oplus \fh$ be a fermion Bogoliubov transformation. There is a unitary transformation $\mathds{U}_{U}: \, \cF^- \ra \cF^-$ such that
\begin{align*}
\mathds{U}_{U} \left[ \fc ( f ) + \fa ( \ol{g} ) \right] \mathds{U}_{U}^* = \fc ( u f + v g ) + \fa ( v \ol{f} + u \ol{g} )
\end{align*}
for all $f \oplus g \in \fh \oplus \fh$ if and only if $v$ is Hilbert--Schmidt. 
\end{lemma}

This condition on $v$ is called Shale-Stinespring condition \cite{SS1}. The proof of this lemma can be found for instance in \cite{Ar1}. 


\subsubsection*{States and Density Matrices}

\begin{definition}
For fermions, states are continuous linear functionals $\omega \in (\cA^-)^*$ on the CAR algebra which are normalized, $ \omega (\1_{\cF}) = 1$, and positive, $\omega (A) \geq 0$ for all positive semi-definite operators $A \in \cA^-$. 
\end{definition}
Since the fermion systems considered in this work --- like atoms and molecules --- are particle number-conserving, we only deal with even states, \ie, for every odd $n \in \NN$, we have
\begin{align*}
\omega ( e (f_1) \cdots e (f_n) ) = 0,
\end{align*}
where $e$ denotes either a creation operator $\fc$ or an annihilation operator $\fa$. Furthermore, we want the particle number expectation value and variance to be finite. Thus, we restrict ourselves to the following subset:

\begin{definition}
We denote the set of all even states with finite particle number variance by
\begin{align*}
\cZ^- = \left\{ \omega \in (\cA^-)^*  \Big| \, \omega \ \textrm{is a state with} \ \omega \big( \hNN^2 \big) < \infty\ \textrm{and} \ \omega \big( e (f_1) \cdots e (f_n) \big) = 0 \ \forall \ \text{odd} \ n \in \NN \right\}.
\end{align*}
Again, $e$ denotes either a creation operator $\fc$ or an annihilation operator $\fa$. 
\end{definition}

\begin{remark}
By the Cauchy--Schwarz inequality, we have for any state $\omega \in \cZ^-$
\begin{align*}
\omega \big( \hNN \big) \leq \sqrt{\omega \big( \hNN^2 \big)} < \infty.
\end{align*}
\end{remark}

\begin{definition}
A state $\omega \in \cZ^-$ is called pure if there is a $\Psi \in \cF^-$, such that 
\begin{align*}
\omega ( A ) = \left< \Psi, A \Psi \right>_{\cF}
\end{align*}
for any $A \in \cA^-$.
\end{definition} 

\begin{definition}
A state $\omega \in \cZ^-$ is called quasifree, shortly $\omega \in \cZ_{\mathrm{qf}}^-$, if it fulfills Wick's Theorem, \ie,
\begin{align}
&\omega ( e_1 \, e_2 \cdots e_{2N-1} ) = 0 \quad \text{and} \notag \\
&\omega ( e_1 \, e_2 \cdots e_{2N} ) = \sum\limits_{\pi} ( -1 )^{\pi} \omega ( e_{\pi (1)} \, e_{\pi (2)} ) \cdots \omega ( e_{\pi (2N-1)} \, e_{\pi (2N)} ) \label{eq: Pfaffian}
\end{align}
for every $N \in \NN$, where $ e_i $ denotes either a creation or an annihilation operator for every $i \in \left\{ 1,2,\dots,2N \right\}$. The sum is taken over all permutations $\pi \in \fS_{2N}$ satisfying 
\begin{align*}
\pi (1) < \pi (3) < \cdots < \pi (2N-1)\
\text{and} \
\pi (2k-1) < \pi (2k)
\end{align*}
for every $k \in \left\{ 1,2,\dots,N \right\}$. The right hand side of \eqref{eq: Pfaffian} is called Pfaffian.

The subset of the pure quasifree states is denoted by $\cZ_{\mathrm{pqf}}^-$. For any $N \in \NN$ and any orthonormal vectors $\varphi_1,\dots,\varphi_N \in \fh$, the vector $\varphi_1 \wedge \dots \wedge \varphi_N \in \cF^-$ is called Slater determinant and defines a pure quasifree state.
\end{definition}

\begin{remark}
The sets $\cZ_{\mathrm{qf}}^-$ and $\cZ_{\mathrm{pqf}}^-$ are invariant under Bogoliubov transformations.
\end{remark}

There is a characterization of pure quasifree states using the Bogoliubov transformation.

\begin{remark}\label{rem: ferm pqf Bogoliubov}
A state $\omega \in \cZ^-$ is pure quasifree if and only if there is a fermion Bogoliubov transformation $U : \fh \oplus \fh \ra \fh \oplus \fh$ with unitary representation $\mathds{U} : \cF^- \ra \cF^-$, such that for any $A \in \cA^-$
\begin{align*}
\omega ( A ) = \left< \mathds{U} \vac, A \mathds{U} \vac \right>_{\cF}.
\end{align*}
\end{remark}

\begin{remark}
Since we assume that the fermion states are even, they are, in particular, centered (see Definition~\ref{def: bos centered state} for bosons). Moreover, for bosons, the set of centered quasifree states is a proper subset of the set of quasifree states (Definition~\ref{def: bos quasifree}), $\cZ_{\mathrm{cqf}}^+ \subsetneq \cZ_{\mathrm{qf}}^+$, while for fermions all quasifree states are centered.
\end{remark}

\begin{definition}
A selfadjoint, positive semi-definite trace class operator $\rho \in \cL^{1} (\cF^-)$ of unit trace, $\tr_{\cF^-} \left( \rho \right) = 1$, is called density matrix.
\end{definition}

The map $ \cA^- \ra \dC, A \mapsto \tr_{\cF^-} \left( \rho^{\frac{1}{2}} A \rho^{\frac{1}{2}} \right)$ defines a state. Since we only study fermion systems that preserve the particle number, we restrict our attention to density matrices which commute with the particle number operator and have a finite squared particle number expectation value,
\begin{align}
 \rho= \bigoplus_{N=0}^{\infty} \rho^{(N)} \qquad \text{and} \qquad \tr_{\cF^-} \left( \rho^{\frac{1}{2}} \, \hNN^2 \, \rho^{\frac{1}{2}} \right) < \infty \label{eq: ferm particle conserving dm}.
\end{align}
Note that, if $m,n\geq0, m \neq n$, then 
\begin{align*}
\tr_{\cF^-} \left\{\rho^{\frac{1}{2}} \, \fc(f_1) \cdots \fc(f_m) \, \fa(g_1) \cdots \fa(g_n) \,\rho^{\frac{1}{2}}\right\} = 0
\end{align*} 
for any choice of $f_1,\dots,f_m,g_1,\dots,g_n\in\fh$, due to \eqref{eq: ferm particle conserving dm}. 

\begin{remark}
In particular, for every state $\omega \in \cZ^-$, there is a density matrix $\rho$ fulfilling \eqref{eq: ferm particle conserving dm} and $\tr_{\cF^-} \left( \rho^{\frac{1}{2}} \, A \, \rho^{\frac{1}{2}} \right) = \omega (A)$ for all $A \in \cA^-$.
\end{remark}


\subsubsection*{One- and Two-Particle Density Matrices}

We now introduce the notion of fermion one- and two-particle density matrices.

\begin{definition}
For any $\omega \in \cZ^-$, the one-particle density matrix (1-pdm) $\gamma_{\omega} \in \cB(\fh)$ of $\omega$ is defined by
\begin{align*}
\left< f , \gamma_{\omega} \, g \right>_{\fh} := \omega \big( \fc(g) \, \fa(f) \big)
\end{align*}
for $f,g \in \fh$.
\end{definition}

\begin{definition}
The two-particle density matrix (2-pdm) $\Gamma_{\omega} : \, \fh \otimes \fh \ra \fh \otimes \fh$ of a state $\omega \in \cZ^-$ is the bounded operator given by
\begin{align*}
\left< f_1 \otimes f_2 , \Gamma_{\omega} \left( g_1\otimes g_2 \right) \right>_{\fh \otimes \fh} := \omega \big( \fc(g_2) \, \fc(g_1) \, \fa(f_1) \, \fa(f_2) \big)
\end{align*}
for $f_1,f_2,g_1,g_2 \in \fh$.
\end{definition}

An outline of basic properties of the fermion 1- and 2-pdm can be found in Lemma~2.1 of \cite{BKM}.


\subsubsection*{Generalized One-Particle Density Matrix for Fermions}

Analogously to the boson case, we define a generalization of the 1-pdm for fermions as in \cite{BLS}. The complex conjugates of a function or of an operator are defined as for bosons in \eqref{eq: bos complex conjugate vector} and \eqref{eq: bos complex conjugate operator}, respectively.

\begin{definition}
Let $\omega \in \cZ^-$ and fix an ONB $\left\{ \varphi_k \right\}_{k=1}^{\infty}$ of $\fh$. Then, the generalized 1-pdm $\tgamma_{\omega}$ of $\omega$ is an operator on $\fh \oplus \fh$ defined by 
\begin{align*}
\left< \left( f_1 \oplus f_2 \right), \tgamma_{\omega} \left( g_1 \oplus g_2 \right) \right> := \omega \big( \big[ \fc ( g_1 ) + \fa ( \ol{g}_2 ) \big] \big[ \fa ( f_1 ) + \fc ( \ol{f}_2 ) \big] \big)
\end{align*}
for $f_1, f_2, g_1, g_2 \in \fh$.
\end{definition} 

\begin{remark}
Again, we define the operator $\alpha_{\omega}^* : \, \fh \ra \fh$ for every $f,g \in \fh$ by
\begin{align*}
\left< f, \alpha_{\omega}^* \, g \right> := \omega \big( \fc ( g ) \, \fc ( \ol{f} ) \big).
\end{align*}
Then, the generalized 1-pdm is expressed as the matrix 
\begin{align*}
\tgamma_{\omega} = \begin{pmatrix} \gamma_{\omega} & \alpha_{\omega} \\ \alpha_{\omega}^* & \1_{\fh} - \overline{\gamma}_{\omega} \end{pmatrix}.
\end{align*}
As for bosons, the fermion 1-pdm $\gamma$ is selfadjoint, but $\alpha$ is anti-symmetric, \ie, $\alpha^T = - \alpha$, as follows from CAR.
\end{remark}

\begin{lemma}\label{lem: ferm bounds on gen 1pdm}
For any $\omega \in \cZ^+$, the generalized 1-pdm $\tgamma_{\omega}$ is a positive semi-definite operator on $\fh \oplus \fh$. In particular, it is selfadjoint. Furthermore, it is bounded above by $\1_{\fh} \oplus \1_{\fh}$.
\end{lemma}

We refer the reader to \cite{BLS} for a proof. From Lemma~\ref{lem: ferm bounds on gen 1pdm} we deduce $0 \leq \gamma \leq \1_{\fh}$ for the 1-pdm.

A consequence of Wick's theorem is the following lemma.
\begin{lemma}
A quasifree state $\omega \in \cZ_{\mathrm{qf}}^-$ is uniquely determined by its generalized 1-pdm $\tgamma_{\omega}$.
\end{lemma}

Furthermore, the generalized 1-pdm transforms in a specific manner under the Bogoliubov transformation.
\begin{lemma}\label{lem: ferm transformed gen 1pdm}
Let $\omega \in \cZ^-$ be a state with generalized 1-pdm $\tgamma : \, \fh \oplus \fh \ra \fh \oplus \fh$. For a fermion Bogoliubov transformation $U : \, \fh \oplus \fh \ra \fh \oplus \fh$ with unitary representation $\mathds{U}_{U} : \, \cF^- \ra \cF^-$ define $\omega_{U}$ by $\omega_{U} ( A ) := \omega \big( \mathds{U}_{U} \, A \, \mathds{U}_{U}^* \big)$ for any $A \in \cA^-$. The generalized 1-pdm $\tgamma_{U}$ corresponding to the state $\omega_{U}$ is given by
\begin{align}
\tgamma_{U} = U^* \, \tgamma \, U.\label{eq: ferm transformed gen 1pdm}
\end{align}
In particular, $\tgamma^2 = \tgamma$ implies $\tgamma_{U}^2 = \tgamma_{U}$.
\end{lemma}

\begin{proof}
For any $f_1, f_2, g_1, g_2 \in \fh$, we have
\begin{align*}
\left< \left( f_1 \oplus f_2 \right), \tgamma_{U} \left( g_1 \oplus g_2 \right) \right> &= \omega \big( \mathds{U}_{U} \left[ \fc ( g_1 ) + \fa ( \ol{g}_2 ) \right] \mathds{U}_{U}^* \mathds{U}_{U} \left[ \fa ( f_1 ) + \fc ( \ol{f}_2 ) \right] \mathds{U}_{U}^* \big) \notag\\
&= \omega \big( \big[ \fc ( u g_1 + v g_2 ) + \fa ( u \ol{g}_2 + v \ol{g}_1 ) \big] \big[ \fa ( u f_1 + v f_2 ) + \fc ( u \ol{f}_2 + v \ol{f}_1 ) \big] \big) \notag \\
&= \left< U \left( g_1 \oplus g_2 \right), \tgamma \, U \left( f_1 \oplus f_2 \right) \right> 
\end{align*}
for the matrix elements of $\tgamma_{U}$. Thus, \eqref{eq: ferm transformed gen 1pdm} holds. Furthermore, by the unitarity of $U$ and \eqref{eq: ferm transformed gen 1pdm}, we obtain $\tgamma_{U}^2 = U^* \, \tgamma \, U \, U^* \, \tgamma \, U = U^* \, \tgamma^2 \, U$ and $\tgamma^2 = \tgamma$ yields $\tgamma_{U}^2 = U^* \, \tgamma \, U = \tgamma_{U}$.
\end{proof}


\subsection{Bogoliubov--Hartree--Fock Theory}

\subsubsection*{Boson Bogoliubov--Hartree--Fock Theory}

For bosons, the number of particles in most physically relevant models is not fixed. As, for instance, in a system of photons interacting with an electron, photons can appear or disappear, depending on what is energetically favorable. Thus, the particle number should not be fixed in the variational process yielding the ground state energy $E_{\gs} := \inf \left\{ \sigma (\Hf) \right\}$. By the Rayleigh--Ritz principle, the ground state energy (as well as the ground state) is determined by
\begin{align*}
E_{\gs} = \inf \left\{ \omega (\Hf) \Big| \, \omega \in \cZ^+ \right\}.
\end{align*}
In the Bogoliubov--Hartree--Fock (BHF) theory, the variation is restricted to quasifree states:
\begin{align*}
E_{\mathrm{BHF}} := \inf \left\{ \omega (\Hf) \Big| \, \omega \in \cZ_{\mathrm{qf}}^+ \right\}.
\end{align*}
The BHF energy $E_{\mathrm{BHF}}$ is an upper bound to the ground state energy $E_{\gs}$. Note that, unlike the common definitions of quasifreeness, our quasifree states are not necessarily centered. Since a quasifree state is uniquely determined by its further generalized 1-pdm $\hgamma_{\omega}$, there is an energy functional $\cE_{\mathrm{BHF}} : \, \cD (\cE_{\mathrm{BHF}}) \ra \dC$, $ \cD (\cE_{\mathrm{BHF}}) \subseteq \cB (\fh \oplus \fh \oplus \dC)$, such that $\omega (\Hf) = \cE_{\mathrm{BHF}} (\hgamma_{\omega})$. Thus, the BHF energy is rewritten as
\begin{align*}
E_{\mathrm{BHF}} = \inf \left\{ \cE_{\mathrm{BHF}} (\hgamma_{\omega}) \Big| \, \omega \in \cZ_{\mathrm{qf}}^+ \right\} = \inf \left\{ \cE_{\mathrm{BHF}} (\hgamma) \Big| \, \hgamma \geq 0, \ \tr (\gamma)<\infty \right\}.
\end{align*}
The second equality is a consequence of two facts: On the one hand, any quasifree state $\omega$ with first moment $b$ is linked to a unique centered quasifree state via the Weyl transformation $\mathds{W}_{b}$. On the other hand, any positive semi-definite operator $\tgamma = \left( \begin{smallmatrix} \gamma & \alpha \\ \alpha^* & \1_{\fh} + \ol{\gamma} \end{smallmatrix} \right)$ on $\fh \oplus \fh$ fulfilling $\tr ( \gamma ) < \infty$ is the generalized 1-pdm of a centered quasifree state, cf. \cite{Na1}.

\subsubsection*{Fermion Bogoliubov--Hartree--Fock Theory}

For fermions, assume $U : \, \dR^3 \ra \dR$ to be an external potential and $V : \, \dR^3 \times \dR^3 \ra \dR_0^+$ a repulsive interaction between two particles. There are multiplication operators associated to these potentials which we also denote by $U$ and $V$, respectively. With the Laplace operator $\Laplace$, the Hamiltonian of the system is given by
\begin{align*}
H_N := \sum\limits_{i=1}^N \left[ - \Laplace_i - U (x_i) \right] + \frac{1}{2} \sum\limits_{1 \leq i,j \leq N} V (x_i,x_j),
\end{align*}
where $x_i \in \dR^3, \ 1 \leq i \leq N$. We only allow for potentials for which $H_N$ is defined as a selfadjoint operator on a dense domain $\cD_N$ and is bounded below. The second quantization of this Hamiltonian is
\begin{align*}
\Hf = \sum\limits_{i,j=1}^{\infty} h_{ij} \fc_i \fa_j + \frac{1}{2} \sum\limits_{i,j,k,l=1}^{\infty} V_{ij,kl} \fc_j \fc_i \fa_k \fa_l,
\end{align*}
where the one-particle operator $h$ and the interaction operator are given by
\begin{align*}
h_{ij} := \left< \varphi_i, \left( - \Laplace - U \right) \varphi_j \right>_{\fh},\\
V_{ij,kl} := \left< \varphi_i \otimes \varphi_j, V \left( \varphi_k \otimes \varphi_l \right) \right>_{\fh \otimes \fh}, 
\end{align*}
respectively, for any elements of a given ONB $\left\{ \varphi_i \right\}_{i=1}^{\infty}$ of $\fh$ with $\sn{ \nabla \varphi_i }_{\fh} < \infty$. The Hamiltonian $H_N$ is the restriction of $\Hf$ to the $N$-particle Fock space $\bigwedge^N \fh$. If we do not assume the dynamics to conserve the particle number, the ground state energy of the $N$-particle system is determined by the Rayleigh--Ritz principle:
\begin{align*}
E_{\gs} = \inf \left\{ \omega (\Hf) \Big| \, \omega \in \cZ^- \right\}.
\end{align*}
Using the energy functional
\begin{align*}
\cE (\gamma,\Gamma) := \tr ( h \gamma ) + \frac{1}{2} \tr ( V \Gamma ),
\end{align*}
this can be re-expressed as
\begin{align*}
E_{\gs} = \inf \left\{ \cE (\gamma,\Gamma) \Big| \, (\gamma,\Gamma) \ \text{is representable} \right\}.
\end{align*}
Here, the problem of representability arises, \ie, a classification of all representable pairs in $\fh \times \fh \otimes \fh$. In order to obtain an upper bound to $E_{\gs}$, the variation is restricted to quasifree states which yields the Bogoliubov--Hartree--Fock energy
\begin{align*}
E_{\mathrm{BHF}} := \inf \left\{ \omega (\Hf) \Big| \, \omega \in \cZ_{\mathrm{qf}}^- \right\} = \inf \left\{ \cE_{\mathrm{BHF}} (\tgamma) \Big| \, \tgamma \geq 0, \ \tr_{\fh} (\gamma) < \infty \right\}.
\end{align*}
For any quasifree state $\omega$, the Bogoliubov--Hartree--Fock functional $\cE_{\mathrm{BHF}}$ is given by $\cE_{\mathrm{BHF}} (\tgamma_{\omega}) := \omega (\Hf)$, where $\tgamma_{\omega}$ is the generalized 1-pdm of $\omega$. 


\section{Bosonic Representability Conditions and the Generalized Two-Particle Density Matrix}\label{sec: Boson representability}

\subsection{Particle Number-Conserving Systems}

To our knowledge, sets of representability conditions given in the literature are for particle number-conserving systems for fermions, as well as for bosons. I.e. only states, that fulfill
\begin{align*}
\omega \left( \left[ \prod\limits_{k=1}^n \bc ( f_k ) \right] \left[ \prod\limits_{l=1}^m \ba ( g_l ) \right] \right) = 0
\end{align*}
for any two sets $\left\{ f_k \right\}_{k=1}^n, \left\{ g_l \right\}_{l=1}^m \subseteq \fh$ with $m,n \in \NN \cup \left\{ 0 \right\}$ and $m \neq n$, are considered.

Since the dynamics of many realistic physical boson systems do not conserve the particle number, an alternative should be found. First, we restate some representability conditions for bosons. 

\begin{definition}
Let $\left( \gamma, \Gamma \right)$ be a pair of operators on $\fh \times \left( \fh \otimes \fh \right)$. We say that $\left( \gamma, \Gamma \right)$ satisfies the representability conditions up to second order with particle number-conservation if
\begin{enumerate}
\item $\left( \gamma, \Gamma \right)$ is admissible,
\item $\Gamma$ satisfies the P-condition, \ie, 
\begin{align*}
\Gamma \geq 0,
\end{align*}
and
\item the G-condition, \ie, for any $A \in \cB ( \fh )$ we have 
\begin{align*}
\tr \left( \left( A^* \otimes A \right) \left[ \Gamma + \Ex \left( \gamma \otimes \1_{\fh} \right) \right] \right) \geq \abs{\tr \left( A \, \gamma \right)}^2.
\end{align*}
\end{enumerate}
\end{definition}

These conditions can be found, \eg, in \cite{GP1,GM1}. Note that these conditions are only necessary conditions, but do not ensure that the considered operators are one- and two-particle density matrices. Furthermore, we omit here other known conditions like the $T_1$- and $T_2$-condition, cf. \cite{Er2}.

\begin{remark}
The Q-condition is omitted since it follows from the P-condition and the positivity of $\gamma$, see \cite{GM1}. Nevertheless, in the same manner, we can rephrase the Q-condition from \cite{GM1} as
\begin{align*}
\Gamma \geq - \left( \1_{\fh \otimes \fh} + \Ex \right) \left( \gamma \otimes \1_{\fh} + \1_{\fh} \otimes \gamma + \1_{\fh} \otimes \1_{\fh} \right).
\end{align*}
\end{remark}

The representability conditions for bosons up to second order are derived in the same spirit as it is done for fermions in \cite{BKM}. 

\begin{theorem}\label{thm: bos repr cond 2nd order w part cons}
Let $\omega$ be a linear continuous functional on $\cA^+$ such that $\omega (\1) = 1$, $\omega \big(\hNN^2\big) < \infty$, and $\omega \big(e_1 \dots e_{2N-1}\big) = 0$ for all $N \in \NN$, where $e_k$ denotes either a creation or annihilation operator. Furthermore, let $\Gamma_{\omega}$ and $\gamma_{\omega}$ be the corresponding 1- and 2-pdm and $\left\{ \varphi_k \right\}_{k=1}^{\infty}$ an ONB of $\fh$. Then the following statements are equivalent:
\begin{itemize}
\item[(i)] For any polynomial $\Po_r \in \cA^+$ in creation and annihilation operators of degree $r \leq 2$, we have
\begin{align*}
\omega ( \Po_r \Po_r^* ) \geq 0.
\end{align*}
\item[(ii)] $\gamma_{\omega} \geq 0$ and $\Gamma_{\omega}$ fulfills the G- and P-condition.
\end{itemize} 
\end{theorem}

Since the proof is analogous to the fermion case considered in \cite{BKM}, we omit the details here. Note that, unlike the fermion case, the trace class conditions on the 1- and 2-pdm cannot be derived from the polynomials since the boson creation and annihilation operators are unbounded.


\subsection{Systems without Particle Number-Conservation and the Generalized Two-Particle Density Matrix}

We generalize the definition of the representability conditions up to second order to systems --- and, thus, states --- which do not conserve the particle number. These representability conditions arise in the same manner as those for particle conserving states by considering expectation values of polynomials up to second order in the creation and annihilation operators. Due to the absence of particle number-conservation, the expectation values of terms, in which the number of creation operators is not equal to the number of annihilation operators, do, in general, not vanish. A simple consequence of Definition~\ref{def: bos gen 2pdm} is the following proposition.
\begin{proposition}
The representability conditions up to second order are satisfied if the pair $\left( \gamma,\Gamma \right)$ of operators on $\fh \times \left( \fh \otimes \fh \right)$ is admissible and $\hGamma$ is positive semi-definite as an operator on $\fH_{\mathrm{sim}}$. 
\end{proposition}

Let $\omega$ be a linear functional on the operators on $\cF^+$. Since, on the one hand, any polynomial up to second order in creation and annihilation operators can be written as $\Po = \Po_2 (F) + \Po_1 (f) + \nu$ with $F \in \bigoplus^4 \fh^{\otimes 2}, \ f \in \bigoplus^2 \fh, \ \nu \in \dC$, Definition~\ref{def: bos gen 2pdm} yields
\begin{align*}
\omega ( \Po \Po^* ) = \left< \left( \begin{matrix} F \\ f \\ \nu \end{matrix} \right) , \hGamma \begin{pmatrix} F \\ f \\ \nu \end{pmatrix} \right>.
\end{align*}
On the other hand, every element of $\fH_{\mathrm{sim}}$ can be written as a vector with $F \in \bigoplus^4 \fh^{\otimes 2}, \ f \in \bigoplus^2 \fh, \ \nu \in \dC$. Thus, the representability conditions up to second order are exactly those arising from
\begin{align*}
\omega \big( \Po \Po^* \big) \geq 0
\end{align*}
for any polynomial $\Po$ in creation and annihilation operators of degree $r \leq 2$.

\begin{remark}
Since the generalized 1-pdm appears as a block in the generalized 2-pdm, it inherits the definiteness property from the generalized 2-pdm.
\end{remark}

If one varies only over particle number-conserving states, then $\hGamma$ assumes a block-diagonal form, and the complexity of the representability reduces considerably. In fact, only three independent conditions remain which are reminiscent of the G- and P-condition in quantum chemistry (see Theorem \ref{thm: bos repr cond 2nd order w part cons}).


\section{Variation over Pure Quasifree States and Bogoliubov--Hartree--Fock Energy}\label{sec: Variation}

For bosons, Theorem~I.2 of \cite{BBT} states for the Pauli--Fierz model that the Bogoliubov--Hartree--Fock energy coincides with the infimimum of the energy functional for a variation over pure quasifree states. We prove a more general statement, which holds for bosons, as well as for fermions. The main result of this section is the following
\begin{theorem}\label{thm: Variation}
Assume the Hamiltonian $\Hf$ to be bounded below. Then,
\begin{align*}
E_{\mathrm{BHF}} = \inf \left\{ \omega \big(\Hf\big) \Big| \, \omega \ \text{is pure and quasifree} \right\} =: E_{\mathrm{BHF}}^{\mathrm{pure}}.
\end{align*}
\end{theorem}
We show the statement in the following two subsections for bosons and fermions separately.

\subsection{Bosons}\label{sec: bos Variation}

For bosons, a more precise statement of Theorem~\ref{thm: Variation} is:
\begin{theorem}\label{thm: BHF-pqf}
Let $\Hf$ be a Hamiltonian on $\cF^+$ that is bounded below. Then,
\begin{align*}
E_{\mathrm{BHF}} = \inf \left\{ \omega \big(\Hf\big) \Big| \, \omega \in \cZ_{\mathrm{pqf}}^+ \right\} =: E_{\mathrm{BHF}}^{\mathrm{pure}}.
\end{align*}
\end{theorem}

In order to prove the theorem, we need some properties of quasifree and pure quasifree states. To this end, we give a characterization of quasifree and pure quasifree states using the Bogoliubov transformation.

\begin{lemma}\label{lem: bos characterize qf}
For any quasifree density matrix, there are a positive semi-definite operator $C \in \cL^1 (\fh)$ with $\sn{C}_{\cB(\fh)} < 1$ and second quantization $\Gamma(C) := \bigoplus_{N=0}^{\infty} C^{\otimes N}$, a boson Bogoliubov transformation with unitary implementation $\mathds{U}$, and $f \in \fh$ such that
\begin{align*}
\rho = \mathds{W}_f \, \mathds{U} \, \frac{\Gamma(C)}{\tr_{\cF^+} (\Gamma(C))} \, \mathds{U}^* \, \mathds{W}_f^*.
\end{align*}
If, furthermore, the density matrix is pure, it is of the form $\mathds{W}_f \, \mathds{U} \ket{\vac} \bra{\vac} \mathds{U}^* \, \mathds{W}_f^*$, where we used the Dirac bra-ket notation.
\end{lemma}

This Lemma is a consequence of Lemma~III.1 in \cite{BBT}.

\begin{proof}[Proof of Theorem~\ref{thm: BHF-pqf}]
Without loss of generality, we assume the Hamiltonian to be positive semi-definite. If $\Hf$ is bounded below, there is a constant $\mu \geq 0$ such that $\Hf_{0} := \Hf + \mu \1_{\cF} \geq 0$. Considering $\Hf_0$ instead of $\Hf$ just adds the constant $\mu$ to both $E_{\mathrm{BHF}}$ and $E_{\mathrm{BHF}}^{\mathrm{pure}}$.

The inequality
\begin{align}
E_{\mathrm{BHF}} = \inf \left\{ \omega \big(\Hf\big) \Big| \, \omega \in \cZ_{\mathrm{qf}}^+ \right\} \leq \inf \left\{ \omega \big(\Hf\big) \Big| \, \omega \in \cZ_{\mathrm{pqf}}^+ \right\} = E_{\mathrm{BHF}}^{\mathrm{pure}}
\end{align} 
follows from the definition of the BHF energy, since the variation is restricted to the proper subset $\cZ_{\mathrm{pqf}}^+ \subsetneq \cZ_{\mathrm{qf}}^+$. 

It remains to prove that $\omega \big(\Hf\big) \geq E_{\mathrm{BHF}}^{\mathrm{pure}}$ for any quasifree state $\omega \in \cZ_{\mathrm{qf}}^+$. Let $\omega \in \cZ_{\mathrm{qf}}^+$ with $\omega \big(\Hf\big) < \infty$ and denote the corresponding density matrix by $\rho$. Then, 
\begin{align*}
\omega \big(\Hf\big) = \tr_{\cF^+} \left( \rho^{\frac{1}{2}} \, \Hf \, \rho^{\frac{1}{2}} \right) = \tr_{\cF^+}\left( \left( \rho^{\frac{1}{2}} \, \Hf^{\frac{1}{2}} \right) \left( \Hf^{\frac{1}{2}} \, \rho^{\frac{1}{2}} \right) \right)
\end{align*}
since $\Hf \geq 0$. Therefore, $\rho^{\frac{1}{2}} \, \Hf^{\frac{1}{2}}$ is Hilbert--Schmidt and we obtain by the cyclicity of the trace
\begin{align*}
\tr_{\cF^+}\left( \left( \rho^{\frac{1}{2}} \, \Hf^{\frac{1}{2}} \right) \left( \Hf^{\frac{1}{2}} \, \rho^{\frac{1}{2}} \right) \right) = \tr_{\cF^+} \left( \Hf^{\frac{1}{2}} \, \rho \, \Hf^{\frac{1}{2}} \right).
\end{align*}
Since $\Hf$ is selfadjoint, there is an ONB $\left\{ \Psi_k \right\}_{k =1}^{\infty}$ of $\cF^+$, such that $\Psi_k \in \cD (\Hf)$ for any $k \in \NN$. Then,
\begin{align*}
\tr_{\cF^+} \left( \Hf^{\frac{1}{2}} \, \rho \, \Hf^{\frac{1}{2}} \right) = \sum\limits_{k=1}^{\infty} \left< \Hf^{\frac{1}{2}} \Psi_k, \rho \, \Hf^{\frac{1}{2}} \Psi_k \right>_{\cF}.
\end{align*}
By Lemma~\ref{lem: bos characterize qf}, the positive semi-definite operator $\rho$ can be written as $\rho = \kappa \kappa^*$, where
\begin{align*}
\kappa := \mathds{W}_f \, \mathds{U} \, \frac{ \Gamma (C^{\frac{1}{2}}) }{ \left( \tr_{\cF^+} (\Gamma (C)) \right)^{\frac{1}{2}} }
\end{align*}
with some $f \in \fh$, a Bogoliubov transformation with unitary implementation $\mathds{U}$, and some $C \in \cL^1 (\fh), C \geq 0, \sn{C}_{\cB(\fh)} < 1$. Hence,
\begin{align}\label{eq: bos split state}
\omega ( \Hf ) = \sum\limits_{k=1}^{\infty} \sn{ \kappa^* \, \Hf^{\frac{1}{2}} \Psi_k}_{\cF}^2.
\end{align}

We continue by introducing a resolution of the identity with coherent states. To this end, we consider an increasing sequence of $n$-di\-men\-sio\-nal Hilbert spaces $\fh_n \subseteq \fh_{n+1} \subseteq \fh, \ n \in \NN$ with $\ol{\bigcup_{n \in \NN} \fh_n} = \fh$ and $C \, \fh_n \subseteq \fh_n$. For any $n$-dimensional Hilbert space $\fh_n$, there is an isometric isomorphism $I : \, \fh_n \ra \dC^n$. We define the measure $\di \mu_n \big(z^{(n)}\big)$ on $\fh_n$ by $\int_{\fh_n} \di \mu_n \big(z^{(n)}\big) f \big(z^{(n)} \big) := \int_{\dC^n} \frac{\di^n x \di^n y}{\pi^n} f \big( I z^{(n)} \big)$, where $x := \Re (z), \ y := \Im (z)$. For any $n \in \NN$, we have $\fh = \fh_n \oplus \fh_n^{\perp}$, where $\fh_n^{\perp}$ denotes the orthogonal complement of $\fh_n$ in $\fh$. Moreover, $\cF^+ \cong \cF^+ [\fh_n] \otimes \cF^+ [\fh_n^{\perp}]$ and $\vac = \vac_n \otimes \vac_n^{\perp}$ with $\vac_n \in \cF^+ [\fh_n]$ and $\vac_n^{\perp} \in \cF^+ [\fh_n^{\perp}]$. For every $n \in \NN$, the projections $\ket{\mathds{W} \big(z^{(n)}\big) \vac}\bra{\mathds{W} \big(z^{(n)}\big) \vac}$, $z^{(n)} \in \fh_n$, satisfy
\begin{align*}
\1_{\cF^+ [\fh_n]} \otimes \ket{\vac_n^{\perp}} \bra{\vac_n^{\perp}} = \int_{\fh_n} \di \mu_n \big( z^{(n)} \big) \ket{\mathds{W} \big(z^{(n)}\big) \, \vac} \bra{\mathds{W} \big(z^{(n)}\big) \, \vac},
\end{align*}
see, \eg, \cite{Be1,CR1}. Consequently,
\begin{align*}
\left< \Psi,\Psi \right>_{\cF} = \lim_{n \ra \infty} \int_{\fh_n}\di \mu_n \big(z^{(n)}\big) \abs{\left< \Psi , \mathds{W} \big(z^{(n)}\big) \, \vac \right>_{\cF}}^2
\end{align*}
for any $\Psi \in \cF^+$. Thus, each summand of the right hand side of \eqref{eq: bos split state} is rewritten as
\begin{align*}
\sn{ \kappa^* \, \Hf^{\frac{1}{2}} \Psi_k }_{\cF}^2 = \lim_{n \ra \infty} \int_{\fh_n} \di \mu_n \big(z^{(n)}\big) \abs{ \left< \Hf^{\frac{1}{2}} \Psi_k, \kappa \, \mathds{W} \big(z^{(n)}\big) \, \vac \right>_{\cF} }^2.
\end{align*}
The sequence $\left( k \mapsto \int_{\fh_n} \di \mu_n \big(z^{(n)}\big) \abs{ \left< \Hf^{\frac{1}{2}} \Psi_k, \kappa \, \mathds{W} \big( z^{(n)} \big) \, \vac \right>_{\cF} }^2 \right)_{n=1}^{\infty}$ is monotonously increasing. Therefore, the summation and the limit is exchanged by the monotone convergence theorem, where the summation is considered as an integral with the counting measure. Thus, we get
\begin{align*}
\omega \big( \Hf \big) =  \lim_{n \ra \infty} \sum\limits_{k=1}^{\infty} \int_{\fh_n} \di \mu_n \big(z^{(n)}\big) \abs{ \left< \Hf^{\frac{1}{2}} \Psi_k, \kappa \, \mathds{W} \big( z^{(n)} \big) \, \vac \right>_{\cF} }^2.
\end{align*}
Afterwards, Fubini's Theorem yields
\begin{align*}
\omega \big( \Hf \big) &= \lim_{n \ra \infty} \int_{\fh_n} \di \mu_n \big(z^{(n)}\big) \sum\limits_{k=1}^{\infty} \abs{ \left< \Hf^{\frac{1}{2}} \Psi_k, \kappa \, \mathds{W} \big( z^{(n)} \big) \, \vac \right>_{\cF} }^2\\
               &= \lim_{n \ra \infty} \int_{\fh_n} \di \mu_n \big(z^{(n)}\big) \left< \kappa \, \mathds{W} \big( z^{(n)} \big) \, \vac, \Hf \, \kappa \mathds{W} \big( z^{(n)} \big) \, \vac \right>_{\cF}
\end{align*} 
and we conclude from the proof of Lemma III.7 in \cite{BBT} that
\begin{align*}
\kappa \, \mathds{W} \big( z^{(n)} \big) \, \vac = \mathds{W}_f \, \mathds{U} \, \frac{ \Gamma (C^{\frac{1}{2}}) }{ \left( \tr_{\cF^+} (\Gamma (C)) \right)^{\frac{1}{2}} } \, \mathds{W} \big( z^{(n)} \big) \, \vac = \nu_{C} \big(z^{(n)} \big) \, \mathds{W}_g \, \mathds{U} \, \vac
\end{align*} 
for some $g \in \fh$ and $\nu_{C} \big( z^{(n)} \big) \in \dC$ with $ \lim_{n \ra \infty} \int_{\fh_n} \di \mu_n \big( z^{(n)} \big) \abs{ \nu_C \big( z^{(n)} \big) }^2 = 1$. By Lemma~\ref{lem: bos characterize qf}, this vector defines a pure quasifree state and, consequently,
\begin{align*}
\omega ( \Hf ) &= \lim_{n \ra \infty} \int_{\fh_n} \di \mu_n \big(z^{(n)}\big) \left< \kappa \, \mathds{W} \big( z^{(n)} \big) \, \vac, \Hf \, \kappa \, \mathds{W} \big( z^{(n)} \big) \vac \right>_{\cF}\\
               &\geq E_{\mathrm{BHF}}^{\mathrm{pure}} \lim_{n \ra \infty} \int_{\fh_n} \di \mu_n \big( z^{(n)} \big) \abs{ \nu_C \left( z^{(n)} \right) }^2\\
               &= E_{\mathrm{BHF}}^{\mathrm{pure}},
\end{align*}
which completes the proof.
\end{proof}


\subsection{Fermions}\label{sec: ferm Variation}

For fermions, a similar result to Theorem~\ref{thm: BHF-pqf} holds:

\begin{theorem}\label{thm: HF-pqf}
Let $\Hf$ be a Hamiltonian on $\cF^-$ that is bounded below. Then,
\begin{align*}
E_{\mathrm{BHF}} = \inf \left\{ \omega \big(\Hf\big) \Big| \, \omega \in \cZ_{\mathrm{pqf}}^- \right\} =: E_{\mathrm{BHF}}^{\mathrm{pure}}.
\end{align*}
\end{theorem}

Before we prove this theorem, we need two preparatory lemmas.

\begin{lemma}\label{lem: decomposition qf}
Let $\omega \in \cZ_{\mathrm{qf}}^-$ with density matrix $\rho$. Then, there are a decomposition $\fh = \fh_{S} \oplus^{\perp} \fh_{\Gamma}$ with $n:= \dim (\fh_{S}) < \infty$, a positive semi-definite trace class operator $B \in \cB(\fh_{\Gamma})$, and a fermion Bogoliubov transformation $U$ with unitary implementation $\mathds{U}$, such that
\begin{align}\label{eq: decomposition qf}
\rho = \mathds{U}_{U} \left( \left( \ket{\varphi_1 \wedge \dots \wedge \varphi_n}\bra{\varphi_1 \wedge \dots \wedge \varphi_n} \right) \otimes \frac{\Gamma (B)}{\tr_{\cF^-} \left( \Gamma (B) \right)} \right) \mathds{U}_{U}^*
\end{align}
for any ONB $\left\{ \varphi_k \right\}_{k=1}^n$ of $\fh_{S}$.
\end{lemma}

\begin{proof}
It is known that there are fermion Bogoliubov transformations $U$ such that the generalized 1-pdm of $\rho_{U} := \mathds{U}_{U}^* \, \rho \, \mathds{U}_{U}$ is of the form
\begin{align}\label{eq: Diag gen1pdm}
\tgamma_{U} = \begin{pmatrix} \gamma_{U} & 0 \\ 0 & \1_{\fh} - \ol{\gamma}_{U} \end{pmatrix}
\end{align}
for some $0 \leq \gamma_{U} \leq \1_{\fh}$ with $\tr_{\fh} (\gamma_{U}) < \infty$, see, \eg, \cite[Theorem~2.3]{BLS}. Let $\fh_{S}$ be the eigenspace of $\gamma_{U}$ associated to the eigenvalue 1 with dimension $n < \infty$ and $\fh_{\Gamma}$ its orthogonal complement. Then, $\gamma_{U} = P_{S} + \gamma_{\Gamma}$, where $P_S$ is the orthogonal projection on $\fh_S$ and $\gamma_{\Gamma}$ the restriction of $\gamma_{U}$ to $\fh_{\Gamma}$. Note that $\gamma_{\Gamma}$ satisfies $\fh_S \subseteq \ker (\gamma_{\Gamma})$, $\gamma_{\Gamma} \, \fh_{\Gamma} \subseteq \fh_{\Gamma}$, and $0 \leq \gamma_{\Gamma} \leq \mu \, \1_{\fh}$ for some $0 < \mu < 1$. Let $\varphi_1,\dots,\varphi_n$ be an ONB of $\fh_S$. Moreover, let
\begin{align*}
\rho' := \ket{\varphi_1 \wedge \dots \wedge \varphi_n} \bra{\varphi_1 \wedge \dots \wedge \varphi_n} \otimes \frac{\Gamma (B)}{\tr_{\cF^-} \left( \Gamma (B) \right)}
\end{align*}
with $B := \left( \gamma_{\Gamma}\right) \left(\1_{\fh_{\Gamma}} - \gamma_{\Gamma}\right)^{-1}$. In order to show that $\rho' = \rho_{U}$, it is sufficient to observe that $\rho'$ defines a quasifree state $\omega'$ and $\tgamma_{\omega'} = \tgamma_{U}$ from \eqref{eq: Diag gen1pdm}, since quasifree states are characterized by their generalized 1-pdm, see \cite{BLS}. Note that we implicitly used the decomposition $\cF^- \cong \cF^- [\fh_{S}] \otimes \cF^- [\fh_{\Gamma}]$.
\end{proof}

\begin{remark}
For any positive semi-definite trace class operator $B \in \cB (\fh_{\Gamma})$, there are an ONB $\left\{ \phi_k \right\}_{k=1}^{\infty}$ of $\fh_{\Gamma}$ and coefficients $b_k \geq 0, \ k \in \NN$, such that $B = \sum_{k=1}^{\infty} b_k \ket{\phi_k} \bra{\phi_k}$ and $\sum_{k=1}^{\infty} b_k < \infty$. Thus, 
\begin{align*}
\tr_{\cF^-} \left( \Gamma (B) \right) = \tr_{\cF^-} \left( \bigotimes_{k=1}^{\infty} \Gamma (b_k) \right) = \prod\limits_{k=1}^{\infty} \tr_{\cF^-} \left( \Gamma (b_k) \right) = \prod\limits_{k=1}^{\infty} \left( 1+b_k \right)
\end{align*}
which converges due to $\sum_{k=1}^{\infty} b_k < \infty$. Here, $\Gamma (b_k)$ should be understood as the second quantized operator $\Gamma (b_k \ket{\phi_k} \bra{\phi_k})$ on $\cF^- [\dC \phi_k]$.
\end{remark}

\begin{lemma}\label{lem: qf as convex combination of pqf}
Let $\omega \in \cZ_{\mathrm{qf}}^-$ with density matrix $\rho$. Then, there is a sequence $\left( \rho_k \right)_{k=1}^{\infty}$ of pure quasifree density matrices and $\left( \lambda_k \right)_{k=1}^{\infty} \in \left[0,\infty\right)^{\NN}$ with $\sum_{k=1}^{\infty} \lambda_k < \infty$ such that
\begin{align*}
\left< \Psi_1, \rho \, \Psi_2 \right>_{\cF} = \lim_{n \ra \infty} \left< \Psi_1, \sum\limits_{k=1}^n \lambda_k \, \rho_k \, \Psi_2 \right>_{\cF}
\end{align*}
for any $\Psi_1,\Psi_2 \in \cF^-$. I.e., every quasifree state is a convex combination of pure quasifree states.
\end{lemma}

\begin{proof}
From Lemma~\ref{lem: decomposition qf}, we know that every quasifree density matrix is of the form \eqref{eq: decomposition qf} and we use the notation specified there in the following. We complete $\left\{ \varphi_k \right\}_{k=1}^n$ to an ONB $\left\{ \varphi_k \right\}_{k=1}^{\infty}$ of $\fh$, where $\left\{ \varphi_k \right\}_{k=n+1}^{\infty}$ is an ONB of $\fh_{\Gamma}$. Then,
\begin{align}\label{eq: approximation of inner product by pqf}
\left< \Psi, \Phi \right>_{\cF} = \lim_{N \ra \infty} \lim_{M \ra \infty} \sum\limits_{k=0}^{N} \sum\limits_{1 \leq i_1 < \cdots < i_k \leq M} \left< \Psi, \varphi_{i_1} \wedge \dots \wedge \varphi_{i_k} \right>_{\cF} \left< \varphi_{1_1} \wedge \dots \wedge \varphi_{i_k}, \Phi \right>_{\cF}
\end{align}
for any $\Psi, \Phi \in \cF^-$. Choosing the ONB $\left\{ \varphi_k \right\}_{k=n+1}^{\infty}$ of $\fh_{\Gamma}$ such that $B$ is diagonalized and using \eqref{eq: approximation of inner product by pqf}, we obtain in the weak sense 
\begin{align*}
\kappa^2 &:= \left( \ket{\varphi_1 \wedge \dots \wedge \varphi_n} \bra{\varphi_1 \wedge \dots \wedge \varphi_n} \otimes \frac{\Gamma (B^{\frac{1}{2}})}{\left[\tr_{\cF^-} \left( \Gamma (B) \right)\right]^{\frac{1}{2}}} \right)^2 \\
&= \lim_{M,N \ra \infty} \sum\limits_{k=n+1}^{N} \sum\limits_{n+1 \leq i_1 < \cdots < i_k \leq M} \left( \ket{\varphi_1 \wedge \dots \wedge \varphi_n} \bra{\varphi_1 \wedge \dots \wedge \varphi_n} \right) \\
&\relphantom{.............................} \otimes \left( \frac{\Gamma (B^{\frac{1}{2}})}{\left[\tr_{\cF^-} \left( \Gamma (B) \right)\right]^{\frac{1}{2}}} \ket{ \varphi_{i_1} \wedge \dots \wedge \varphi_{i_k} } \bra{\varphi_{i_1}  \wedge \dots \wedge \varphi_{i_k}} \frac{\Gamma (B^{\frac{1}{2}})}{\left[\tr_{\cF^-} \left( \Gamma (B) \right)\right]^{\frac{1}{2}}} \right).
\end{align*}
This can be written as
\begin{align*}
\kappa^2 &= \frac{1}{\tr_{\cF^-} \left( \Gamma (B) \right)} \lim_{M,N \ra \infty} \sum\limits_{k=n+1}^{N} \sum\limits_{n+1 \leq i_1 < \cdots < i_k \leq M} \left( \ket{\varphi_1 \wedge \dots \wedge \varphi_n} \bra{\varphi_1 \wedge \dots \wedge \varphi_n} \right) \\
&\relphantom{...........................................................} \otimes \ket{B^{\frac{1}{2}} \varphi_{i_1} \wedge \dots \wedge B^{\frac{1}{2}} \varphi_{i_k}} \bra{B^{\frac{1}{2}} \varphi_{i_1} \wedge \dots \wedge B^{\frac{1}{2}} \varphi_{i_k}} \\
&= \frac{1}{\tr_{\cF^-} \left( \Gamma (B) \right)} \lim_{M,N \ra \infty} \sum\limits_{k=n+1}^{N} \, \sum\limits_{n+1 \leq i_1 < \cdots < i_k \leq M} \\
&\relphantom{..............} \ket{\varphi_1 \wedge \dots \wedge \varphi_n \wedge B^{\frac{1}{2}} \varphi_{i_1} \wedge \dots \wedge B^{\frac{1}{2}} \varphi_{i_k}} \bra{\varphi_1 \wedge \dots \wedge \varphi_n \wedge B^{\frac{1}{2}} \varphi_{i_1} \wedge \dots \wedge B^{\frac{1}{2}} \varphi_{i_k}}.
\end{align*}
Each operator $\ket{\varphi_1 \wedge \dots \wedge \varphi_n \wedge B^{\frac{1}{2}} \varphi_{i_1} \wedge \dots \wedge B^{\frac{1}{2}} \varphi_{i_k}} \bra{\varphi_1 \wedge \dots \wedge \varphi_n \wedge B^{\frac{1}{2}} \varphi_{i_1} \wedge \dots \wedge B^{\frac{1}{2}} \varphi_{i_k}}$ is either equal to zero or a pure quasifree density matrix (up to a normalization constant). Finally, a pure quasifree density matrix conjugated by a Bogoliubov transformation is a pure quasifree state, too, which completes the proof.
\end{proof}

Now, we are prepared to prove Theorem~\ref{thm: HF-pqf}. Since the proof is, to a large extend, similar to the proof of Theorem~\ref{thm: BHF-pqf}, we only give details where there are differences.

\begin{proof}[Proof of Theorem~\ref{thm: HF-pqf}]
Again, without loss of generality, we assume that the Hamiltonian is positive semi-definite. 

As for bosons, the inequality
\begin{align*}
E_{\mathrm{BHF}} = \inf \left\{ \omega \big(\Hf\big) \Big| \, \omega \in \cZ_{\mathrm{qf}}^- \right\} \leq \inf \left\{ \omega \big(\Hf\big) \Big| \, \omega \in \cZ_{\mathrm{pqf}}^- \right\} = E_{\mathrm{BHF}}^{\mathrm{pure}}
\end{align*} 
is immediate. 

Thus, we show $\omega \big(\Hf\big) \geq E_{\mathrm{BHF}}^{\mathrm{pure}}$ for any $\omega \in \cZ_{\mathrm{qf}}^-$. Let $\omega \in \cZ_{\mathrm{qf}}^-$ with $\omega \big(\Hf\big) < \infty$ and denote the corresponding density matrix by $\rho$. Furthermore, let $\left\{ \Psi_k \right\}_{k =1}^{\infty}$ be an ONB of $\cF^-$, such that $\Psi_k \in \cD \big(\Hf\big)$ for any $k \in \NN$. Analogously to the boson case, we obtain
\begin{align*}
\tr_{\cF^-} \left( \Hf^{\frac{1}{2}}\, \rho \, \Hf^{\frac{1}{2}} \right) = \sum\limits_{k=1}^{\infty} \left< \Hf^{\frac{1}{2}} \, \Psi_k, \rho \, \Hf^{\frac{1}{2}} \, \Psi_k \right>_{\cF}.
\end{align*}
By Lemma~\ref{lem: decomposition qf}, the positive semi-definite operator $\rho$ can be written as $\rho = \kappa \, \kappa^*$, where 
\begin{align*}
\kappa := \mathds{U}_{U} \left[ \ket{\varphi_1 \wedge \dots \wedge \varphi_n} \bra{\varphi_1 \wedge \dots \wedge \varphi_n} \otimes \frac{\Gamma (B^{\frac{1}{2}})}{\left[\tr_{\cF^-} \left( \Gamma (B) \right)\right]^{\frac{1}{2}}} \right]
\end{align*} 
with a decomposition $\fh = \fh_S \oplus^{\perp} \fh_{\Gamma}$, $n := \dim (\fh_S) < \infty$, an ONB $\left\{ \varphi_k \right\}_{k=1}^n$ of $\fh_S$, a unitarily implementable Bogoliubov transformation $U$, and a positive semi-definite trace class operator $B \in \cB (\fh_{\Gamma})$. Hence,
\begin{align*}
\omega \big( \Hf \big) = \sum\limits_{k=1}^{\infty} \sn{ \kappa^* \, \Hf^{\frac{1}{2}} \, \Psi_k}_{\cF}^2.
\end{align*}

Instead of a resolution of the identity by coherent states for bosons, we use the resolution of the identity by Slater determinants, 
\begin{align*}
\1_{\cF^-} = \lim_{N \ra \infty} \lim_{M \ra \infty} \sum\limits_{l=0}^{N} \sum\limits_{1 \leq i_1 < \cdots < i_l \leq M} \ket{\varphi_{i_1} \wedge \dots \wedge \varphi_{i_l}} \bra{\varphi_{1_1} \wedge \dots \wedge \varphi_{i_l}},
\end{align*}
as in the proof of Lemma~\ref{lem: qf as convex combination of pqf}, in particular, Eq.~\eqref{eq: approximation of inner product by pqf}. Then, we obtain
\begin{align*}
\omega \big(\Hf\big) = \sum\limits_{k=1}^{\infty} \lim_{N \ra \infty} \lim_{M \ra \infty} \sum\limits_{l=0}^{N} \sum\limits_{1 \leq i_1 < \cdots < i_l \leq M} \abs{ \left< \Hf^{\frac{1}{2}} \, \Psi_k, \kappa \left( \varphi_{i_1} \wedge \dots \wedge \varphi_{i_l} \right) \right>_{\cF} }^2.
\end{align*}
Because the sequence $\left( k \mapsto \lim_{M \ra \infty} \sum\limits_{l=0}^{N} \sum\limits_{1 \leq i_1 < \cdots < i_l \leq M} \abs{ \left< \Hf^{\frac{1}{2}} \, \Psi_k, \kappa \left( \varphi_{i_1} \wedge \dots \wedge \varphi_{i_l} \right) \right>_{\cF} }^2 \right)_{N=1}^{\infty}$ is monotonously increasing, the monotone convergence theorem allows for a exchange of the $k$-summation and the first limit. Using the monotone convergence theorem a second time to exchange the second limit and the $k$-summation, we obtain
\begin{align*}
\omega \big( \Hf \big) =  \lim_{N \ra \infty} \lim_{M \ra \infty} \sum\limits_{k=1}^{\infty} \sum\limits_{l=0}^{N} \sum\limits_{1 \leq i_1 < \cdots < i_l \leq M} \abs{ \left< \Hf^{\frac{1}{2}} \, \Psi_k, \kappa \left( \varphi_{i_1} \wedge \dots \wedge \varphi_{i_l} \right) \right>_{\cF} }^2.
\end{align*}
Furthermore, we change the order of the summations, since the sum is absolutely convergent, and get
\begin{align*}
\omega \big( \Hf \big) = \lim_{N \ra \infty} \lim_{M \ra \infty} \sum\limits_{l=0}^{N} \sum\limits_{1 \leq i_1 < \cdots < i_l \leq M}  \left< \kappa \left( \varphi_{i_1} \wedge \dots \wedge \varphi_{i_l} \right), \Hf \, \kappa \left( \varphi_{i_1} \wedge \dots \wedge \varphi_{i_l} \right) \right>_{\cF}.
\end{align*} 
Every vector $\kappa \left( \varphi_{i_1} \wedge \dots \wedge \varphi_{i_l} \right)$ defines a pure quasifree state, cf. the proof of Lemma~\ref{lem: qf as convex combination of pqf}. Since 
\begin{align*}
\left< \Psi, \Hf \Psi \right>_{\cF} \geq E_{\mathrm{BHF}}^{\mathrm{pure}}
\end{align*} 
for any pure quasifree state $\Psi \in \cF^-$, we finally have
\begin{align*}
\omega \big( \Hf \big) &\geq E_{\mathrm{BHF}}^{\mathrm{pure}} \lim_{N \ra \infty} \lim_{M \ra \infty} \sum\limits_{l=0}^{N} \sum\limits_{1 \leq i_1 < \cdots < i_l \leq M} \left< \varphi_{i_1} \wedge \dots \wedge \varphi_{i_l}, \rho \left( \varphi_{i_1} \wedge \dots \wedge \varphi_{i_l} \right) \right>_{\cF}\\
               &= E_{\mathrm{BHF}}^{\mathrm{pure}}.
\end{align*}
This proves the assertion.
\end{proof}

Theorem~\ref{thm: Variation} now follows from the Theorems~\ref{thm: BHF-pqf} and \ref{thm: HF-pqf}.


\section[Pure Quasifree States and their Generalized One-Particle Density Matrix]{Pure Quasifree States and their Generalized One-Particle Density Matrix}\label{sec: Relation pqf-1pdm}

For a given generalized fermion 1-pdm $\tgamma$, it is known that there is a pure quasifree state $\omega$ which has $\tgamma$ as its generalized 1-pdm, if and only if the generalized 1-pdm is a projection, \ie, $\tgamma^2 = \tgamma$  (see Sect.~\ref{sec: Fermion-qf-pdm}). For bosons, a similar statement is also known. In this section, we show that an even stronger relation holds:
\begin{theorem}\label{thm: relation}
The following statements are equivalent:
\begin{enumerate}
\item[(i)] $\omega$ is a centered pure quasifree state.
\item[(ii)] The corresponding generalized 1-pdm $\tgamma$ satisfies $\tr \left( \gamma \right) < \infty$ and
\begin{align*}
\tgamma \, \cS \, \tgamma = - \tgamma \qquad \text{for bosons},\\
\tgamma^2 = \tgamma \qquad \text{for fermions}. 
\end{align*}
\end{enumerate}
\end{theorem}
Recall $\cS = \1_{\fh} \oplus (-\1_{\fh}) \in \cB (\fh \oplus \fh)$. A proof of Theorem~\ref{thm: relation} in the boson case is given in the following subsection. Two consequences of this theorem are discussed afterwards. In the second subsection, we prove the statement for fermions.

\subsection{Bosons}\label{sec: Boson-qf-pdm}

Before we show Theorem~\ref{thm: relation} for bosons, we give some preparatory lemmas.

\begin{lemma}\label{lem: bos relation quasiproj to qf}
If an operator $\tgamma = \left( \begin{smallmatrix} \gamma & \alpha \\ \alpha^* & \1 + \ol{\gamma} \end{smallmatrix} \right) : \, \fh \oplus \fh \ra \fh \oplus \fh$ satisfies $\tgamma \geq 0, \ \tr (\gamma) < \infty$, and
\begin{align}\label{eq: bos relation quasiproj to qf}
\tgamma \, \cS \, \tgamma = - \tgamma,
\end{align}
there is a centered pure quasifree state $\omega \in \cZ_{\mathrm{pqf}}^+ \cap \cZ_{\mathrm{cen}}^+$ that has $\tgamma$ as its generalized one-particle density matrix.
Furthermore, let $\omega \in \cZ_{\mathrm{pqf}}^+ \cap \cZ_{\mathrm{cen}}^+$ be a centered pure quasifree state. Then, the corresponding generalized 1-pdm $\tgamma$ fulfills \eqref{eq: bos relation quasiproj to qf}.
\end{lemma}

For a proof see \eg \cite{Na1,So2}.

Eq.~\eqref{eq: bos relation quasiproj to qf} is rewritten in a single equation for operators on $\fh$, \ie, we do not need the matrices $\tgamma$ and $\cS$.

\begin{proposition}\label{prop: bos reduction quasiproj}
Let $\omega \in \cZ^+$ be a state with the generalized 1-pdm $\tgamma = \left( \begin{smallmatrix} \gamma & \alpha \\ \alpha^* & \1_{\fh} + \ol{\gamma} \end{smallmatrix} \right)$. Then the following statements are equivalent:
\begin{enumerate}
\item[(i)] $\tgamma \, \cS \, \tgamma = - \tgamma$.
\item[(ii)] $\gamma^2 + \gamma = \alpha \, \alpha^*$. 
\end{enumerate}
\end{proposition}

\begin{proof}
Computing and simplfying the matrix products of {\em{(i)}}, we obtain the four equations
\begin{align}
\gamma^2 + \gamma &= \alpha \, \alpha^*,\label{eq: quasiproj1-1} \\
\ol{\gamma}^2 + \ol{\gamma} &= \alpha^* \, \alpha,\label{eq: quasiproj1-2} \\
\gamma \, \alpha &= \alpha \, \ol{\gamma},\label{eq: quasiproj2-1} \\
\alpha^* \, \gamma &= \ol{\gamma} \, \alpha^*.\label{eq: quasiproj2-2}
\end{align}
Thus, the implication {\emph{(i)}} $\Rightarrow$ {\emph{(ii)}} is immediate. It remains to prove {\emph{(ii)}} $\Rightarrow$ {\emph{(i)}}. Eq.~\eqref{eq: quasiproj2-2} is the adjoint of \eqref{eq: quasiproj2-1}, and \eqref{eq: quasiproj1-2} is equivalent to \eqref{eq: quasiproj1-1}. The system of equations reduces to \eqref{eq: quasiproj1-1} and \eqref{eq: quasiproj2-1}. Furthermore, we show that \eqref{eq: quasiproj2-1} follows from \eqref{eq: quasiproj1-1}. We define $f : \, \dR^+ \ra \dR^+$ by $f(y) := \sqrt{y+1/4}-1/2$ and observe that $f$ is the inverse map of $x \mapsto x + x^2 \ , \ \dR^+ \ra \dR^+$. Then
\begin{align*}
\gamma = f ( \alpha \, \alpha^* ) \ \text{and} \ \ol{\gamma} = f ( \ol{\alpha \, \alpha^*} ) = f ( \alpha^* \alpha ).
\end{align*}
Since $\alpha \, \alpha^*$ is bounded, we approximate the function $f$ by a sequence $\left( p_n \right)_{n=1}^{\infty}$ of polynomials, \ie, $\lim_{n \ra \infty} p_n ( x ) = f ( x )$ uniformly on the compact interval $\left[ - \sn{\alpha \, \alpha^*}_{\mathrm{op}},\sn{\alpha \, \alpha^*}_{\mathrm{op}} \right] \subset \dR$. So, $p_n ( \alpha \, \alpha^* )$ and $p_n ( \alpha^* \, \alpha )$ are well-defined. Using $\left( \alpha \, \alpha^* \right)^m \alpha = \alpha \left( \alpha^* \alpha \right)^m $ for all $m \in \NN$ and limits in operator norm, we obtain
\begin{align*}
\gamma \, \alpha = f \left( \alpha \, \alpha^* \right) \alpha = \lim_{n \ra \infty} p_n \left( \alpha \, \alpha^* \right) \alpha = \lim_{n \ra \infty} \alpha \, p_n \left( \alpha^* \alpha \right) = \alpha \, f \left( \alpha^* \alpha \right)= \alpha \, \ol{\gamma},
\end{align*}
which proves the assertion.
\end{proof}

\begin{remark}
In \cite{BBT}
\begin{align*}
\gamma = \frac{1}{2} \left( \cosh ( 2r ) - \1 \right), \quad
\hat{\alpha} = \frac{1}{2} \sinh ( 2r )
\end{align*}
are used, where $\left< f, \alpha \, g \right>_{\fh} = \left< f, \hat{\alpha} \, \ol{g} \right>_{\fh}$ and $r : \, \fh \ra \fh$ is an antilinear operator. $r$ obeys $\left< f, r \, g \right> = \left< g, r \, f \right>$ for any $f,g \in \fh$ and $r^2$ is trace class. These two equations are, however, implied by \eqref{eq: quasiproj1-1} and, in turn, yield \eqref{eq: quasiproj2-1}.
\end{remark}

Centered pure quasifree states can be characterized by a Bogoliubov transformation (see \cite{Na1} for the proof):

\begin{lemma}\label{lem: cpqf-Bogoliubov}
A centered boson state $\omega \in \cZ_{\mathrm{cen}}^+$ is pure quasifree if and only if there is a boson Bogoliubov transformation $U : \, \fh \oplus \fh \ra \fh \oplus \fh$ with unitary representation $\mathds{U}_{U} : \, \cF^+ \ra \cF^+$, such that for any $A \in \cA^+$
\begin{align*}
\omega ( A ) = \left< \mathds{U}_{U} \, \vac, A \, \mathds{U}_{U} \, \vac \right>.
\end{align*}
\end{lemma}

The relation between a generalized 1-pdm fulfilling \eqref{eq: bos relation quasiproj to qf} and the corresponding centered pure quasifree state is even closer. 

\begin{lemma}\label{lem: bos uniqueness quasiproj to cpqf}
Let $\omega \in \cZ_{\mathrm{cen}}^+$ and assume that the corresponding generalized 1-pdm $\tgamma$ satisfies
\begin{align}
\tgamma \, \cS \, \tgamma = - \tgamma\label{eq: bos quasiproj}.
\end{align}
Then, $\omega$ is a centered pure quasifree state.
\end{lemma}

\begin{proof} 
As stated in Remark~\ref{rem: bos gen 1pdm matrix}, the generalized 1-pdm is of the form $\tgamma = \left( \begin{smallmatrix} \gamma & \alpha \\ \alpha^* & \1 + \ol{\gamma} \end{smallmatrix} \right)$, where $\gamma : \, \fh \ra \fh$ is the 1-pdm and $\alpha^* : \, \fh \ra \fh$ is defined in \eqref{eq: bos alpha}. For any $k \in \NN$, let $\phi_k$ denote an eigenfunction corresponding to the eigenvalue $\lambda_k$ of $\gamma$, \ie, $\gamma \phi_k = \lambda_k \phi_k$. Since the 1-pdm is selfadjoint, we choose a set of eigenfunctions $\left\{ \phi_k \right\}_{k=1}^{\infty}$ which forms an ONB of $\fh$. We define the operators $u : \, \fh \ra \fh$ and $v : \, \fh \ra \fh$ by 
\begin{align*}
u \phi_k := \left( \1 + \gamma \right)^{\frac{1}{2}} \phi_k := \left( 1 + \lambda_k \right)^{\frac{1}{2}} \phi_k \qquad \text{and} \qquad v \phi_k := \alpha \left( \1 + \ol{\gamma} \right)^{-\frac{1}{2}} \phi_k := \left( 1 + \lambda_k \right)^{-\frac{1}{2}} \alpha \phi_k,
\end{align*}
where we abbreviate $\1 \equiv \1_{\fh}$. We show that
\begin{align*}
U := \begin{pmatrix} u & v \\ \ol{v} & \ol{u} \end{pmatrix} = \begin{pmatrix} \left( \1 + \gamma \right)^{\frac{1}{2}} & \alpha \left( \1 + \ol{\gamma} \right)^{-\frac{1}{2}} \\ \alpha^* \left( \1 + \gamma \right)^{-\frac{1}{2}} & \left( \1 + \ol{\gamma} \right)^{\frac{1}{2}} \end{pmatrix}
\end{align*}
defines a boson Bogoliubov transformation. To this end, we prove the conditions on $u$ and $v$ specified in Definition~\ref{def: bosBogoliubov}. We know from the proof of Proposition~\ref{prop: bos reduction quasiproj} that \eqref{eq: bos quasiproj} is equivalent to \eqref{eq: quasiproj1-1}--\eqref{eq: quasiproj2-2}. Using $\alpha \, \ol{\gamma} = \gamma \, \alpha$ and $\alpha \, \alpha^* = \gamma + \gamma^2$ we calculate
\begin{align*}
u u^* - v v^* = \1 + \gamma - \left( \1 + \gamma \right)^{-1} \alpha \, \alpha^* = \1,
\end{align*}
which is the left equation of \eqref{eq: def bosBogoliubov 1}. The right equation of \eqref{eq: def bosBogoliubov 1} is derived from
\begin{align*}
u^* u - v^T \ol{v} = \left( \1 + \gamma \right)^{\frac{1}{2}} \left( \1 + \gamma \right)^{\frac{1}{2}} - \left( \1 + \gamma \right)^{-\frac{1}{2}} \alpha^T \, \ol{\alpha} \left( \1 + \gamma \right)^{-\frac{1}{2}} 
\end{align*}
and using $\alpha^T = \alpha$ and $\alpha \, \alpha^* = \gamma^2 + \gamma$.
Furthermore,
\begin{align*}
u^* v - v^T \ol{u} = \left( \1 + \gamma \right)^{\frac{1}{2}} \alpha \left( \1 + \ol{\gamma} \right)^{-\frac{1}{2}} - \left( \1 + \gamma \right)^{-\frac{1}{2}} \alpha \left( \1 + \ol{\gamma} \right)^{\frac{1}{2}} = 0
\end{align*}
because $\alpha^T = \alpha$ and $\alpha \, \ol{\gamma} = \gamma \, \alpha$. Thus, we get the left equation of \eqref{eq: def bosBogoliubov 2}. Analogously, we obtain
\begin{align*}
u v^T - v u^T = \left( \1 + \gamma \right)^{\frac{1}{2}} \left( \1 + \gamma \right)^{-\frac{1}{2}} \alpha^T - \alpha \left( \1 + \ol{\gamma} \right)^{-\frac{1}{2}} \left( \1 + \ol{\gamma} \right)^{\frac{1}{2}} = 0.
\end{align*}
Hence, $U$ is a boson Bogoliubov transformation. Since
\begin{align*}
\tr \left( v^* v \right) = \tr \left( \left( \1 + \ol{\gamma} \right)^{-\frac{1}{2}} \alpha^* \alpha \left( \1 + \ol{\gamma} \right)^{-\frac{1}{2}} \right) = \tr \left( \ol{\gamma} \right) = \tr \left( \gamma \right) = \omega \big( \hNN \big) < \infty,
\end{align*}
there is a unitary implementation $\mathds{U} : \, \cF^+ \ra \cF^+$ of the Bogoliubov transformation $U$ by Lemma~\ref{lem: bosBogUnitaryImplementation}.

We define a state $\omega_{U} \in \cZ^+$ by $\omega_{U} (A) := \omega (\mathds{U} \, A \, \mathds{U}^*)$ for any $A \in \cA^+$. We show that the generalized 1-pdm of $\omega_{U}$ is 
\begin{align}
\tgamma_{U} = \begin{pmatrix} 0 & 0 \\ 0 & \1 \end{pmatrix}.\label{eq: bos diagonal gen 1pdm}
\end{align}
Since $\tgamma_{U} = U^* \, \tgamma \, U$ by Lemma~\ref{lem: bos transformed 1pdm}, \eqref{eq: bos diagonal gen 1pdm} is equivalent to
\begin{align*}
\begin{pmatrix} \gamma & \alpha \\ \alpha^* & \1 + \ol{\gamma} \end{pmatrix} = U \begin{pmatrix} 0 & 0 \\ 0 & \1 \end{pmatrix} U^* = \begin{pmatrix} v v^* & v u^T \\ \ol{u} v^* & \ol{u} u^T \end{pmatrix}.
\end{align*}
Thus, we only check that $\gamma = v v^*$ and $\alpha = v u^T$. On the one hand,
\begin{align*}
v v^* = \alpha \left( \1 + \ol{\gamma} \right)^{-1} \alpha^* = \left( \1 + \gamma \right)^{-1} \alpha \, \alpha^* = \gamma
\end{align*}
with $\alpha \, \ol{\gamma} = \gamma \, \alpha$ and $\alpha \, \alpha^* = \gamma + \gamma^2$. On the other hand,
\begin{align*}
v u^T = \alpha \left( \1 + \ol{\gamma} \right)^{-\frac{1}{2}} \left( \1 + \ol{\gamma} \right)^{\frac{1}{2}} = \alpha.
\end{align*}
Therefore, the Bogoliubov transformation $U$ yields $\tgamma_{U} = \left( \begin{smallmatrix} 0 & 0 \\ 0 & \1 \end{smallmatrix} \right)$. In particular, we have
\begin{align}\label{eq: bos vanishing 1pdm}
\gamma_{U} = 0
\end{align}
and $\tgamma_{U} \, \cS \, \tgamma_{U} = - \tgamma_{U}$.

Next, we show that the only state having $\tgamma_{U}$ as its generalized 1-pdm is the vacuum state. We choose $\left\{ \varphi_n \right\}_{n=1}^{\infty}$ to be a fixed, but arbitrary ONB of $\fh$. We consider the set
\begin{align*}
\cK := \left\{ K : \, \NN \ra \NN \cup \left\{ 0 \right\}  \Big| \, \exists n_0 \in \NN, \ \forall n \geq n_0 : \, K_n = 0 \right\}
\end{align*}
and define for any $K \in \cK$
\begin{align*}
\bc_K := \prod\limits_{ \substack{n=1\\K_n \neq 0} }^{\infty} \frac{\left(\bc_n\right)^{K_n}}{\sqrt{K_n!}}
\end{align*}
and $\Psi_K \in \cF^+$ by $\Psi_K := \bc_K \vac$. Note that  $\Psi_{\left( 0,0,\dots \right)} = \vac$. For any $K,L \in \cK$ with $K \neq L$, $\left< \Psi_K , \Psi_L \right> = 0$, $\left< \Psi_K , \Psi_K \right> = 1$, and $\left\{ \Psi_K \right\}_{K \in \cK} \subset \cF^+$ forms an ONB of the boson Fock space, the so-called occupancy number basis. 

We denote by $\rho_{U}$ the density matrix corresponding to $\omega_{U}$. With the occupancy number representation and the usual Dirac bra-ket notation, we rewrite the density matrix $\rho_{U}$ of the state $\omega_{U}$ as
\begin{align*}
\rho_{U} = \sum\limits_{K,L \in \cK} \mu_{K,L} \ket{ \Psi_K } \bra{ \Psi_L },
\end{align*}
where $\mu_{K,L} := \omega_{U} ( \ket{\Psi_L} \bra{\Psi_K} ) \in \dC$. For any $n \in \NN$ and any $K \in \cK$ with $K_n \geq 1$, we denote the vector, in which one of the particles in the state given by $\varphi_n$ is removed, by $ K - E^{(n)}$, where $E^{(n)} \in \cK$ with $E_n^{(n)} = 1$ and $E_m^{(n)} = 0$ for all $m \in \NN, \ m \neq n$. Then, for all $n \in \NN$ and $K,L \in \cK$ with $K_n \geq 1$, the Cauchy--Schwarz inequality yields
\begin{align*}
\abs{\mu_{K,L}}^2 &= \abs{ \omega_{U} \big( \ket{\Psi_L} \bra{\Psi_{K - E^{(n)} }} \frac{1}{\sqrt{K_n}} \ba_n \big)}^2\notag \\ 
&\leq \frac{1}{K_n} \omega_{U} \big( \ket{\Psi_L} \left< \Psi_{K - E^{(n)} } , \Psi_{K - E^{(n)} } \right> \bra{\Psi_L} \big) \, \omega_{U} ( \bc_n \ba_n ),
\end{align*}
which vanishes since, by \eqref{eq: bos vanishing 1pdm}, $\omega_{U} ( \bc_n \ba_n ) = \left< \varphi_n, \gamma_{U} \, \varphi_n \right> = 0$ for all $n \in \NN$. Consequently, $\mu_{K,L} = 0$ if any $K,L \in \cK$ is different from $\left( 0,0,\dots \right)$, and $\mu_{K,K} = 1$ for $K = \left( 0,0,\dots \right)$. As asserted, we have
\begin{align*}
\omega_{U} ( A ) = \left< \vac , A \, \vac \right>
\end{align*}
for any $A \in \cA^+$. Hence, we obtain
\begin{align*}
\omega ( A ) = \omega_{U} ( \mathds{U}^* A \, \mathds{U} ) = \left< \mathds{U} \, \vac , A \, \mathds{U} \, \vac \right>
\end{align*}
for the original state $\omega$. So, $\omega$ is a pure quasifree state according to Lemma~\ref{lem: cpqf-Bogoliubov}.
\end{proof}


We conclude from Lemmas~\ref{lem: bos relation quasiproj to qf} and \ref{lem: bos uniqueness quasiproj to cpqf}:

\begin{theorem}
Let $\omega \in \cZ_{\mathrm{cen}}^+$ be a centered state and $\cS = \1_{\fh} \oplus (-\1_{\fh}) \in \cB \left( \fh \oplus \fh \right)$. Then, the following statements are equivalent:
\begin{enumerate}
\item[(i)] $\omega$ is a pure quasifree state.
\item[(ii)] The generalized one-particle density matrix $\tgamma$ of the state $\omega$ satisfies $\tr \left( \gamma \right) < \infty$ and
\begin{align*}
&\tgamma \, \cS \, \tgamma = - \tgamma. 
\end{align*}
\end{enumerate}
\end{theorem}

\begin{proof}
The implication {\em{(i)}} $\Rightarrow$ {\em{(ii)}} is given by the second assertion of Lemma~\ref{lem: bos relation quasiproj to qf} and the reverse by Lemma~\ref{lem: bos uniqueness quasiproj to cpqf}.
\end{proof}

A consequence of Lemma~\ref{lem: bos uniqueness quasiproj to cpqf} is the following corollary. 
\begin{corollary}
Let $\omega \in \cZ^+$ be a state and 
\begin{align*}
\hgamma = \begin{pmatrix} \gamma & \alpha & b \\ \alpha^* & \1 + \ol{\gamma} & \ol{b} \\ b^* & \ol{b}^* & 1  \end{pmatrix}
\end{align*}
the corresponding further generalized 1-pdm with $\gamma : \fh \ra \fh$ and $\alpha^* : \fh \ra \fh$ as defined in Eqs.~\eqref{eq: bos 1pdm} and \eqref{eq: bos alpha}, respectively. As in \eqref{eq: bos first moment}, the first moment $b \in \fh$ of the state $\omega$ is given by $\left< f , b \right> := \omega \big( \ba ( f ) \big)$ for any $f \in \fh$. Furthermore, we define the selfadjoint operator $\cQ_f  : \, \fh \oplus \fh \oplus \dC \ra \fh \oplus \fh \oplus \dC$ by
\begin{align*}
\cQ_f := \begin{pmatrix} \1_{\fh} & 0 & -f \\ 0 & -\1_{\fh} & \ol{f} \\ -f^* & \ol{f}^* & -1 \end{pmatrix}
\end{align*}
for any $f \in \fh$. If
\begin{align}
\hgamma \, \cQ_b \, \hgamma = - \hgamma, \label{eq: bos generalized quasiproj}
\end{align}
then $\omega$ is a pure quasifree state.
\end{corollary}

\begin{proof}
If $\omega$ is centered, we have $b = 0$ and \eqref{eq: bos generalized quasiproj} reduces to $\hgamma \, \cQ_0 \, \hgamma = - \hgamma$ which is equivalent to \eqref{eq: bos quasiproj}. So, Lemma~\ref{lem: bos uniqueness quasiproj to cpqf} directly yields the assertion. Now, we do not assume the state to be centered. Then, for the Weyl transformation $\mathds{W}_b : \, \cF^+ \ra \cF^+$, we define the state $\omega_0 \in \cZ^+$ by $\omega_0 ( A ) := \omega \big( \mathds{W}_b^* A \, \mathds{W}_b \big)$ for any $A \in \cA^+$. First, we show that $b_0 = 0$, $\gamma_0 = \gamma - \ket{b} \bra{b}$, and $\alpha_0^* = \alpha^* - \big|\ol{b}\big>\big<{b}\big|$ for this state $\omega_0$. For any $f \in \fh$, we have
\begin{align*}
\left< b_0 , f \right> :=\omega_0 \big( \bc ( f ) \big) = \omega \big( \bc ( f ) - \left< b,f \right> \1_{\cF} \big) = \left< b,f \right> - \left< b,f \right> = 0.
\end{align*}
Thus, $b_0 = 0$ and $\omega_0$ is a centered state. Furthermore, for any $f,g \in \fh$,
\begin{align*}
\left< f, \gamma_0 g \right> &:= \omega_0 \big( \bc ( g ) \, \ba ( f ) \big) \notag \\
                             &= \omega \big( \big[ \bc (g) - \left< b,g \right> \1_{\cF} \big] \big[ \ba ( f ) - \left< f,b \right> \1_{\cF} \big] \big)\notag \\
                             &= \left< f , \gamma g \right> - \left< f,b \right> \left< b,g \right>.
\end{align*}
An analogous calculation yields
\begin{align*}
\left< \ol{f} , \alpha_0^* g \right> &:= \omega_0 \big( \bc ( g ) \, \bc ( f ) \big) \notag \\
                                     &= \omega \big( \big[ \bc ( g ) - \left< b,g \right> \1_{\cF} \big] \big[ \bc ( f ) - \left< b,f \right> \1_{\cF} \big] \big) \notag \\
                                     &= \omega \big( \bc ( g ) \, \bc ( f ) \big) - \left< b,g \right> \left< b,f \right>.
\end{align*}
Next, we consider \eqref{eq: bos generalized quasiproj}. For every $f \in \fh$ we decompose the operator $\cQ_f$  as
\begin{align*}
\cQ_f = \cR_f \, \widehat{S} \, \cR_f^*,
\end{align*}
where the operators $\cR_f, \widehat{\cS} : \, \fh \oplus \fh \oplus \dC \ra \fh \oplus \fh \oplus \dC$ are given by
\begin{align*}
\cR_f := \begin{pmatrix} -\1_{\fh} & 0 & 0 \\ 0 & -\1_{\fh} & 0 \\ f^* & \ol{f}^* & 1 \end{pmatrix} \quad \text{and} \quad \widehat{\cS} := \begin{pmatrix} \1_{\fh} & 0 & 0 \\ 0 & -\1_{\fh} & 0 \\ 0 & 0 & -1 \end{pmatrix}.
\end{align*}
Since $\cR_b$ is invertible, \eqref{eq: bos generalized quasiproj} is equivalent to
\begin{align}
\left( \cR_b^* \, \hgamma \, \cR_b \right) \widehat{\cS} \left( \cR_b^* \, \hgamma \, \cR_b \right) = - \cR_b^* \, \hgamma \, \cR_b. \tag{\ref{eq: bos generalized quasiproj}'}\label{eq: version bos gen quasiproj}
\end{align}
A straightforward computation yields
\begin{align*}
\cR_b^* \, \hgamma \, \cR_b = \begin{pmatrix} \gamma - \ket{b}\bra{b} & \alpha - \ket{b}\bra{\ol{b}} & 0 \\ \alpha^* - \ket{\ol{b}}\bra{b} & \1_{\fh} + \ol{\gamma} - \ket{\ol{b}}\bra{\ol{b}} & 0 \\ 0 & 0 & 1 \end{pmatrix}.
\end{align*}
Summarising the results, we obtain
\begin{align*}
\cR_b^* \, \hgamma \, \cR_b = \begin{pmatrix} \gamma_0 & \alpha_0 & 0 \\ \alpha_0^* & \1_{\fh} + \ol{\gamma}_0 & 0 \\ 0 & 0 & 1 \end{pmatrix},
\end{align*}
which is the further generalized 1-pdm $\hgamma_0$ of the state $\omega_0$. Thus, \eqref{eq: bos generalized quasiproj} implies
\begin{align*}
\hgamma_0 \, \widehat{\cS} \, \hgamma_0 = - \hgamma_0.
\end{align*}
Since the upper left $2 \times 2$-matrix of $\hgamma_0$ (which is an operator on $\fh \oplus \fh$) is the generalized 1-pdm $\tgamma_0$ and $\widehat{\cS}$ is diagonal, we find
\begin{align*}
\tgamma_0 \, \cS \, \tgamma_0 = - \tgamma_0.
\end{align*}
Hence, the generalized 1-pdm $\tgamma_0$ fulfills \eqref{eq: bos quasiproj} and for the state $\omega_0$ the requirements of Theorem~\ref{lem: bos uniqueness quasiproj to cpqf} are satisfied. Therefore, $\omega_0$ is a pure quasifree state.
\end{proof}

An important set of states which are related to the vacuum state via a Weyl transformation is the set of coherent states. Recall that a state $\omega \in \cZ^+$ is called coherent if there is an $f \in \fh$ and a Weyl transformation $\mathds{W}_f : \, \cF^+ \ra \cF^+$, such that, for any $A \in \cA^+$,
\begin{align*}
\omega \left( A \right) = \left< \mathds{W}_f^* \, \vac, A \, \mathds{W}_f^* \, \vac \right>.
\end{align*}

\begin{corollary}
 Eq.~\eqref{eq: bos generalized quasiproj} is satisfied for every coherent state.
\end{corollary}

\begin{proof}
For every coherent state $\omega$, we can find $\phi \in \fh$, such that
\begin{align*}
\omega \left( A \right) = \left< \vac , \mathds{W}_{\phi} \, A \, \mathds{W}_{\phi}^* \, \vac \right>,
\end{align*}
where $\mathds{W}_{\phi} : \, \cF^+ \ra \cF^+$ is a Weyl transformation.

We have $b = \phi$, because, for any $f \in \fh$,
\begin{align*}
\left< b , f \right> = \left< \vac , \mathds{W}_{\phi} \, \bc ( f ) \, \mathds{W}_{\phi}^* \, \vac \right> = \left< \vac , \big[ \bc ( f ) + \left< \phi , f \right> \1 \big] \vac \right>,
\end{align*}
where we use $\left< \vac , \bc ( f ) \vac \right> = \left< \ba ( f ) \vac , \vac \right> = 0$. For the 1-pdm $\gamma$ we find
\begin{align*}
\left< f , \gamma \, g \right> = \left< \vac , \big[ \bc ( g ) + \left< \phi , g \right> \big] \big[ \ba ( f ) + \left< f , \phi \right> \1 \big] \vac \right> = \left< f , \phi \right> \left< \phi , g \right> = \left< f, b \right> \left< b , g \right>
\end{align*}
for every $f,g \in \fh$ and, thus, $\gamma = \ket{b} \bra{b}$. Furthermore, $\alpha^* = \big| \ol{b} \big>  \big< b \big|$ by
\begin{align*}
\left< \ol{f} , \alpha^* g \right> = \left< \vac , \big[ \bc ( g ) + \left< \phi , g \right> \1 \big] \big[ \bc ( f ) + \left< \phi , f \right> \1 \big] \vac \right> = \left< \phi , g \right> \left< \phi , f \right> = \left< \ol{f} , \ol{b} \right> \left< b , g \right>.
\end{align*}
Finally, we obtain
\begin{align*}
\cR_b^* \, \hgamma \, \cR_b = \begin{pmatrix} 0 & 0 & 0 \\ 0 & \1_{\fh} & 0 \\ 0 & 0 & 1 \end{pmatrix},
\end{align*}
which obviously fulfills \eqref{eq: version bos gen quasiproj}. 
\end{proof}


\subsection{Fermions}\label{sec: Fermion-qf-pdm}

The statements of Sect.~\ref{sec: Boson-qf-pdm} can also be transferred to fermion systems. The fermion analogue of Lemma~\ref{lem: bos relation quasiproj to qf} is the following lemma.

\begin{lemma}\label{lem: ferm relation proj to qf}
If an operator $\tgamma = \left( \begin{smallmatrix} \gamma & \alpha \\ \alpha^* & \1 - \ol{\gamma} \end{smallmatrix} \right) : \, \fh \oplus \fh \ra \fh \oplus \fh$ satisfies $0 \leq \tgamma \leq \1_{\fh}, \ \tr (\gamma) < \infty$, and
\begin{align}\label{eq: ferm proj}
\tgamma^2 = \tgamma,
\end{align}
then there is a unique pure quasifree state $\omega \in \cZ^-$ that has $\tgamma$ as its generalized one-particle density matrix.
Furthermore, let $\omega \in \cZ^-$ be a pure quasifree state. Then, the corresponding generalized 1-pdm $\tgamma$ fulfills \eqref{eq: ferm proj}.
\end{lemma}

This lemma is a consequence of Theorems~2.3 and 2.6 of \cite{BLS}.

\begin{proof}
From \cite[Theorem~2.3]{BLS} we conclude that, for every generalized 1-pdm $\tgamma$ , there is a unique quasifree state $\omega \in \cZ^-$ having $\tgamma$ as its generalized 1-pdm. On the one hand, \cite[Theorem~2.6]{BLS} implies that this quasifree state is pure since the corresponding generalized 1-pdm is a projection. This proves the first assertion of the lemma.

On the other hand, \cite[Theorem~2.6]{BLS} also states that the generalized 1-pdm of a pure quasifree state is a projection which is the second assertion and completes the proof.
\end{proof}

There is even a one-to-one relation between pure quasifree states and generalized 1-pdms fulfilling \eqref{eq: ferm proj}. 

\begin{lemma}\label{lem: ferm uniqueness proj to pqf}
Let $\omega \in \cZ^-$. If the generalized 1-pdm $\tgamma$ corresponding to the state $\omega$ satisfies
\begin{align}
\tgamma^2 = \tgamma\label{eq: ferm proj to pqf},
\end{align}
then $\omega$ is a pure quasifree state.
\end{lemma}

\begin{proof}
Let $\1 \equiv \1_{\fh}$ and $\tgamma = \left( \begin{smallmatrix} \gamma & \alpha \\ \alpha^* & \1 - \ol{\gamma} \end{smallmatrix} \right)$ be the generalized 1-pdm of $\omega$. Eq.~\eqref{eq: ferm proj to pqf} implies $\alpha \, \ol{\gamma} = \gamma \, \alpha$, $\alpha^* \gamma = \ol{\gamma} \, \alpha^*$, $\alpha \, \alpha^* = \gamma - \gamma^2$, and $\alpha^* \alpha = \ol{\gamma} - \ol{\gamma}^2$. We denote by $\left\{ \lambda_i \right\}_{i=1}^{\infty}$ the eigenvalues of the 1-pdm $\gamma$ (counting also degeneracies) and choose the corresponding eigenfunctions $\phi_i \in \fh$, $i \in \NN$, in such a way that $\left\{ \phi_i \right\}_{i=1}^{\infty}$ is an ONB of the one-particle Hilbert space $\fh$. Furthermore, let $P : \, \fh \ra \fh$ be the orthogonal projection on the eigenspace of the eigenvalue 1 of $\gamma$ and $P^{\perp} := \1_{\cF} - P$ the projection orthogonal to $P$. Note that both projections commute with the 1-pdm and that $P$ is also the projection on the eigenspace of the eigenvalue 1 of $\ol{\gamma}$. Furthermore, from $\alpha \, \ol{\gamma} = \gamma \, \alpha$ we obtain $\alpha \, P \, \fh \subseteq P \, \fh$ and $\alpha \, P^{\perp} \fh \subseteq P^{\perp} \fh$. So, $P$ and $P^{\perp}$ commute with $\gamma,\ \ol{\gamma}, \ \alpha$, and $\alpha^*$.

We define $\left( \1 - \gamma \right)^{\frac{1}{2}}$ and $\left( \1 - \ol{\gamma} \right)^{-\frac{1}{2}} P^{\perp}$ by
\begin{align*}
\left( \1 - \gamma \right)^{\frac{1}{2}} \phi_i = \left( \1 - \lambda_i \right)^{ \frac{1}{2}} \phi_i \ \text{and} \left( \1 - \gamma \right)^{-\frac{1}{2}} P^{\perp} \phi_i = \left( 1 - \lambda_i \right)^{-\frac{1}{2}} P^{\perp} \phi_i
\end{align*}
and consider a Bogoliubov transformation $U : \, \fh \oplus \fh \ra \fh \oplus \fh, \ U = \left( \begin{smallmatrix} u & v \\ \ol{v} & \ol{u} \end{smallmatrix} \right)$, given by
\begin{align*}
u := \left( \1 - \gamma \right)^{\frac{1}{2}} \qquad \text{and} \qquad v := \alpha \left( \1 - \ol{\gamma} \right)^{-\frac{1}{2}} P^{\perp} + P.
\end{align*}
First we show that $U$ is indeed a Bogoliubov transformation, \ie, that (\ref{eq: def ferm Bogoliubov}a,b) hold. The operators $u$ and $v$ satisfy
\begin{align*}
u u^* + v v^* &= \left( \1 - \gamma \right) + \left[ \alpha \left( \1 - \ol{\gamma} \right)^{-\frac{1}{2}} P^{\perp} + P \right] \left[ \left( \1 - \ol{\gamma} \right)^{-\frac{1}{2}} P^{\perp} \alpha^* + P \right] \notag \\
&= \1 - \gamma + \alpha \left( \1 - \ol{\gamma} \right)^{-1} P^{\perp} \alpha^* + P + \alpha \left( \1 - \ol{\gamma} \right)^{-\frac{1}{2}} P^{\perp} P + P \left( \1 - \ol{\gamma} \right)^{-\frac{1}{2}} P^{\perp} \alpha^*.
\end{align*}
With $\alpha^* \gamma = \ol{\gamma} \, \alpha^*$ and $\alpha \, \alpha^* = \gamma - \gamma^2$, we have
\begin{align*}
u u^* + v v^* = \1 - \gamma + \alpha \, \alpha^* \left( \1 - \gamma \right)^{-1} P^{\perp} + \gamma P = \1.
\end{align*}
$u^* u + v^T \ol{v} = \1$ can be shown analogously. Moreover, using $\left( \1 - \gamma \right)^{\frac{1}{2}} P = 0$ and $\left( \1 - \ol{\gamma} \right)^{\frac{1}{2}} P = 0$,
\begin{align*}
u^* v + v^T \ol{u} &= \left( \1 - \gamma \right)^{\frac{1}{2}} \left[ \alpha \left( \1 - \ol{\gamma} \right)^{-\frac{1}{2}} P^{\perp} + P \right] + \left[ \left( \1 - \ol{\gamma} \right)^{-\frac{1}{2}} P^{\perp} \alpha^T + P \right] \left( \1 - \ol{\gamma} \right)^{\frac{1}{2}} \notag \\
&= \alpha \, P^{\perp} + \left( \1 - \gamma \right)^{\frac{1}{2}} P - P^{\perp} \alpha  + P \left( \1 - \ol{\gamma} \right)^{\frac{1}{2}} \notag \\
&=0.
\end{align*}
$u v^T + v u^T = 0$ is obtained similarly. Since, furthermore, the (operator valued) entries on the diagonal of the matrix $U$, and, due to $\alpha^* = - \ol{\alpha}$, those on the off-diagonal as well, are complex conjugate to each other, $U$ is a fermion Bogoliubov transformation according to Definition~\ref{def: ferm Bogoliubov}. 

The Bogoliubov transformation $U$ has a unitary representation $\mathds{U}$ because
\begin{align*}
\tr \left( v^* v \right) &= \tr \left( \left[ \left( \1 - \ol{\gamma} \right)^{-\frac{1}{2}} P^{\perp} \alpha^* + P \right] \left[ \alpha \left( \1 - \ol{\gamma} \right)^{-\frac{1}{2}} P^{\perp} + P \right] \right) \notag \\
                         &= \tr \left( \left( \1 - \ol{\gamma} \right)^{-\frac{1}{2}} \ol{\gamma} \left( \1 - \ol{\gamma} \right) \left( \1 - \ol{\gamma} \right)^{-\frac{1}{2}} P^{\perp} + P \right)
\end{align*}
due to $\alpha^* \alpha = \ol{\gamma} - \ol{\gamma}^2$ and, thus,
\begin{align*}
\tr \left( v^* v \right) &= \tr \left( \ol{\gamma} P^{\perp} + P \right) = \tr \left( \gamma \right) = \omega ( \hNN ) < \infty.
\end{align*}

We define a state $\omega_{U} \in \cZ^-$ by $\omega_{U} ( A ) := \omega \big( \mathds{U} \, A \, \mathds{U}^* \big)$ for any $A \in \cA^-$ and denote its density matrix by $\rho_{U}$. We show that the corresponding generalized 1-pdm $\tgamma_{U}$ is given by
\begin{align*}
\tgamma_{U} = \begin{pmatrix} 0 & 0 \\ 0 & \1 \end{pmatrix}.
\end{align*}
Consequently,
\begin{align}
\gamma_{U} = 0 \label{eq: ferm vanishing 1pdm}
\end{align}
and the transformed generalized 1-pdm $\tgamma_{U}$ is a projection. 

By Lemma~\ref{lem: ferm transformed gen 1pdm}, the Bogoliubov transformation $U$ yields $\tgamma_{U} = U^* \, \tgamma \, U$. So, $U$ satisfies
\begin{align*}
\tgamma = U \begin{pmatrix} 0 & 0 \\ 0 & \1 \end{pmatrix} U^* = \begin{pmatrix} v v^* & v u^T \\ \ol{u} v^* & \ol{u} u^T \end{pmatrix},
\end{align*}
that is $\gamma = v v^*$ and $\alpha = v u^T$. This is indeed the case since
\begin{align*}
v v^* &= P + \alpha \left( \1 - \ol{\gamma} \right)^{-1} P^{\perp} \alpha^* = \gamma P + \alpha \, \alpha^* \left( \1 - \gamma \right)^{-1} P^{\perp} = \gamma P + \gamma P^{\perp} = \gamma, \\
v u^T &= \alpha \left( \1 - \ol{\gamma} \right)^{-\frac{1}{2}} P^{\perp} \left( \1 - \ol{\gamma} \right)^{\frac{1}{2}} + \alpha \left( \1 - \ol{\gamma} \right)^{-\frac{1}{2}} P^{\perp} P = \alpha P^{\perp} = \alpha.
\end{align*} 
Here, we used $\alpha P = 0$ which follows from $P \alpha^* \alpha P = P \left( \ol{\gamma} - \ol{\gamma}^2 \right) P = 0$.

Let $\left\{ \varphi_n \right\}_{n=1}^{\infty}$ denote an arbitrary ONB of $\fh$. We define
\begin{align*}
\cK := \left\{ k \in \NN \ra \left\{ 0,1 \right\} \Big| \, \exists n_0 \in \NN \ \forall n \geq n_0 : \, k_n = 0 \right\}.
\end{align*}
The elements $K \in \cK$ are the occupancy number representations of the fermion Fock space. If we define
\begin{align*}
\fc_K := \prod\limits_{\substack{n=1\\K_n \neq 0}}^{\infty} \fc_n
\end{align*}
for any $K \in \cK$, the functions $\Psi_{\left(0,0,\dots\right)} = \vac$ and $\Psi_K \in \cF^-$, given by $\Psi_K := \fc_K \vac$ for $K \in \cK, \ K \neq \left( 0,0,\dots \right)$, form an ONB of the fermion Fock space. Furthermore, for any $n \in \NN$ and any $K \in \cK$ with $K_n = 1$, we write $K \setminus \left\{ n \right\}$ for the set where the particle in the state given by $\varphi_n$ is removed, but the others are left unchanged. Now, we write the density matrix corresponding to $\omega_{U}$ as
\begin{align*}
\rho_{U} = \sum\limits_{K,L \in \cK} \mu_{K,L} \ket{\Psi_K} \bra{\Psi_L},
\end{align*} 
where the coefficients are given by $ \mu_{K,L} := \omega_{U} \big( \ket{\Psi_L} \bra{\Psi_K} \big) \in \dC$ for any sets $K,L \in \cK$. Applying the Cauchy--Schwarz inequality, we obtain for every $n \in \NN$ and every pair $K,L \in \cK$ with $K_n = 1$
\begin{align*}
\abs{ \mu_{K,L} }^2 &= \abs{\omega_{U} \big( \fc_L \ket{\vac} \bra{\vac} \fa_{K \setminus \left\{ n \right\}} \fa_n \big)}^2 \leq \omega_{U} \big( \fc_L \ket{\vac} \bra{\vac} \fa_{K \setminus \left\{ n \right\}} \fc_{K \setminus \left\{ n \right\}} \ket{\vac} \bra{\vac} \fa_L \big) \, \omega_{U} \big( \fc_n \fa_n \big).
\end{align*}
By \eqref{eq: ferm vanishing 1pdm}, $\omega_{U} \big( \fc_n \fa_n \big) = \left< \varphi_n, \gamma_{U} \, \varphi_n \right> = 0$ for every $n \in \NN$ and $\mu_{K,L} = 0$ if one of the sets $K,L \in \cK$ is not $\left( 0,0,\dots \right)$. Hence, for any $A \in \cA^-$
\begin{align*}
\omega_{U} ( A ) = \left< \vac, A\, \vac \right>.
\end{align*}
Since $U$ is a Bogoliubov transformation with unitary implementation $\mathds{U}$ and invertible, we obtain for any $A \in \cA^-$
\begin{align*}
\omega ( A ) &= \omega_{U} \big( \mathds{U}^* A \, \mathds{U} \big) = \left< \mathds{U} \, \vac, A \, \mathds{U} \, \vac \right>.
\end{align*}
Therefore, the state $\omega$ is pure and quasifree by Remark~\ref{rem: ferm pqf Bogoliubov} which yields the assertion.
\end{proof}

From the last two lemmas we conclude:

\begin{theorem}
Let $\omega \in \cZ^-$ be a state. Then the following statements are equivalent:
\begin{enumerate}
\item[(i)] $\omega$ is a pure quasifree state.
\item[(ii)] The generalized 1-pdm $\tgamma$ of the state $\omega$ satisfies
\begin{align*}
\tgamma^2 = \tgamma. 
\end{align*}
\end{enumerate}
\end{theorem}

\begin{proof}
The implication {\em{(i)}} $\Rightarrow$ {\em{(ii)}} is given by the second assertion of Lemma~\ref{lem: ferm relation proj to qf} and the reverse by Lemma~\ref{lem: ferm uniqueness proj to pqf}.
\end{proof}






\appendix
\section{Generalized 2-pdm as $7 \times 7$-Matrix}\label{sec: appendix}

In this appendix, we give a more explicit, basis dependent form of the generalized 2-particle density matrix $\hGamma$. We assume $\left\{ \phi_k \right\}_{k=1}^{\infty}$ to be a fixed, but arbitrary ONB of $\fh$. Recall that $\hGamma$ is defined as a $7 \times 7$-matrix on $\fh_{\mathrm{sim}}$ in Definition~\ref{def: bos gen 2pdm}. In order to simplify notation, it is convenient to define some operators and functionals. 

\begin{definition}
Let $f_1, f_2, g_1, g_2, f, g \in \fh$, and $\mu \in \dC$. We define $\cD (B) := \big\{ F \in \fh \otimes \fh \big| \, \sum_{k=1}^{\infty} \left< \phi_k \otimes \phi_k, F \right>_{\fh \otimes \fh} < \infty \big\} \subseteq \fh \otimes \fh$ and the following linear maps:
\begin{align*}
&\Lambda_1 : \, \fh \otimes \fh \to \fh \otimes \fh \ , \ \left< g_1 \otimes g_2, \Lambda_1 \left( f_1 \otimes f_2 \right) \right> := \omega \left( \bc ( f_1 ) \bc ( f_2 ) \bc ( \ol{g}_2 ) \ba ( g_1 ) \right),\\
&\Lambda_2^*: \, \fh \otimes \fh \to \fh \otimes \fh \ , \ \left< g_1 \otimes g_2, \Lambda_2^* \left( f_1 \otimes f_2 \right) \right> := \omega \left( \bc ( f_1 ) \bc ( f_2 ) \bc ( \ol{g}_2 ) \bc ( \ol{g}_1 ) \right),\\
&\Delta : \, \fh \otimes \fh \to \fh \otimes \fh \ , \ \left< g_1 \otimes g_2, \Delta \left( f_1 \otimes f_2 \right) \right> := \omega \left( \bc ( f_1 ) \bc ( \ol{g}_1 ) \ba ( g_2 ) \ba ( \ol{f}_2 ) \right),\\
&A_1: \, \fh \otimes \fh \to \fh \ , \ \left< g, A_1 \left( f_1 \otimes f_2 \right) \right> := \omega \left( \bc ( f_1 ) \bc ( f_2 ) \ba ( g ) \right),\\
&A_2^* : \, \fh \otimes \fh \to \fh \ , \ \left< g, A_2^* \left( f_1 \otimes f_2 \right) \right> := \omega \left( \bc ( f_1 ) \bc ( f_2 ) \bc ( \ol{g} ) \right),\\
&Q_1 : \, \fh \otimes \fh \to \fh \ , \ \left< g, Q_1 \left( f_1 \otimes f_2 \right) \right> := \omega \left( \bc ( f_1 ) \ba ( \ol{f}_2 ) \ba ( g ) \right),\\
&Q_2 : \, \fh \otimes \fh \to \fh \ , \ \left< g, Q_2 \left( f_1 \otimes f_2 \right) \right> := \omega \left( \bc ( f_1 ) \ba ( \ol{f}_2 ) \bc ( \ol{g} ) \right),\\
&B : \, \cD (B) \to \fh \otimes \fh \ ,\ B:= \sum\limits_{i,k=1}^{\infty} \left| \phi_i \otimes \phi_i \right> \left< \phi_k \otimes \phi_k \right|,\\
&\beta_2: \, \dC \otimes \fh \to \fh \ , \ \beta_2 := \sum\limits_{i=1}^{\infty} \left| \phi_i \right> \left< 1 \otimes \phi_i \right|,\\
&\beta_1 :\, \cD (B) \to \dC \ , \ \beta_1 := \sum\limits_{i=1}^{\infty} \left< \phi_i \otimes \phi_i \right|.
\end{align*}
\end{definition}

Furthermore, recall from Remark~\ref{rem: bos further gen 1pdm matrix} that $b$ is given by $\left< b , f \right>_{\fh} := \omega \left( \bc ( f ) \right)$.

As we already pointed out in Proposition~\ref{prop: gen 2-pdm is selfadjoint and nonnegative}, the generalized 2-pdm is selfadjoint. Since therefore $\hGamma_{ij} = \ol{\hGamma}_{ji}$, it suffices to state the entries $\hGamma_{ij}$ for $1 \leq i \leq j \leq 7$. Using the notation specified before, we have:
\allowdisplaybreaks
\begin{align*}
&\hGamma_{11} = \Gamma : \, \fh \otimes \fh \ra \fh \otimes \fh,\\
&\hGamma_{12} = \Lambda_1^* : \, \fh \otimes \fh \ra \fh \otimes \fh,\\
&\hGamma_{13} = \Lambda_1^* \Ex + \left( \alpha \otimes \1 \right) B : \, \fh \otimes \fh \ra \fh \otimes \fh,\\
&\hGamma_{14} = \Lambda_2 : \, \fh \otimes \fh \ra \fh \otimes \fh,\\
&\hGamma_{15} = A_1^* : \, \fh \ra \fh \otimes \fh,\\
&\hGamma_{16} = A_2 : \, \fh \ra \fh \otimes \fh,\\
&\hGamma_{17} = \left( \alpha \otimes \1 \right) \beta_1^* : \, \dC \ra \fh \otimes \fh,\\
&\hGamma_{22} = \Ex \Delta + \gamma \otimes \1 : \, \fh \otimes \fh \ra \fh \otimes \fh,\\
&\hGamma_{23} = \Delta^* + \left(\gamma \otimes \1 \right) \left(  B + \Ex \right) : \, \fh \otimes \fh \ra \fh \otimes \fh,\\
&\hGamma_{24} = \Ex \ol{\Lambda}_1 + \left( \alpha \otimes \1 \right) \left(1+\Ex \right) : \, \fh \otimes \fh \ra \fh \otimes \fh,\\
&\hGamma_{25} = Q_1^* : \, \fh \ra \fh \otimes \fh,\\
&\hGamma_{26} = Q_2^* : \, \fh \ra \fh \otimes \fh,\\
&\hGamma_{27} = \left( \1 \otimes \gamma \right) \beta_1^* : \, \dC \ra \fh \otimes \fh,\\
&\hGamma_{33} = \Delta \Ex + \left( \1 \otimes \gamma \right) B + B \left( \1 \otimes \gamma \right) + B + \1 \otimes \gamma : \, \fh \otimes \fh \ra \fh \otimes \fh,\\
&\hGamma_{34} = \ol{\Lambda}_1 + \left( \1 \otimes \alpha \right) \left( 1 + \Ex \right) + B \left( \1 \otimes \alpha \right): \, \fh \otimes \fh \ra \fh \otimes \fh,\\
&\hGamma_{35} = \Ex Q_1^* + \beta_1^* b^* : \, \fh \ra \fh \otimes \fh,\\
&\hGamma_{36} = \Ex Q_2^* + \beta_1^* \ol{b}^* : \, \fh \ra \fh \otimes \fh,\\
&\hGamma_{37} = \left( \1 \otimes \1 + \1 \otimes \gamma \right) \beta_1^* : \, \dC \ra \fh \otimes \fh,\\
&\hGamma_{44} = \Gamma^T + \left( \1 \otimes \1 + \ol{\gamma} \otimes \1 + \1 \otimes \ol{\gamma} \right) \left( 1 + \Ex \right) : \, \fh \otimes \fh \ra \fh \otimes \fh,\\
&\hGamma_{45} = \ol{A}_2 : \, \fh \ra \fh \otimes \fh,\\
&\hGamma_{46} = \ol{A}_1^* + \left( 1+\Ex \right) \left( \1 \otimes b \right) \beta_2^* : \, \fh \ra \fh \otimes \fh,\\
&\hGamma_{47} = \left( \alpha^* \otimes \1 \right) \beta_1^* : \, \dC \ra \fh \otimes \fh,\\
&\hGamma_{55} = \gamma : \, \fh \ra \fh,\\
&\hGamma_{56} = \alpha : \, \fh \ra \fh,\\
&\hGamma_{57} = b : \, \dC \ra \fh,\\
&\hGamma_{66} = \1 + \ol{\gamma} : \, \fh \ra \fh,\\
&\hGamma_{67} = \ol{b} : \, \dC \ra \fh,\\
&\hGamma_{77} = 1 : \, \dC \ra \dC.
\end{align*}

\begin{remark}
Both generalizations of the 1-pdm, given in Definitions~\ref{def: bos gen 1-pdm} and \ref{def: bos further gen 1-pdm}, are contained in the generalized 2-pdm, namely
\begin{align*}
\tgamma = \begin{pmatrix} \hGamma_{55} & \hGamma_{65} \\ \hGamma_{56} & \hGamma_{66} \end{pmatrix} \quad \text{and} \quad \hgamma = \begin{pmatrix} \hGamma_{55} & \hGamma_{65} & \hGamma_{75} \\ \hGamma_{56} & \hGamma_{66} & \hGamma_{76} \\ \hGamma_{57} & \hGamma_{67} & \hGamma_{77} \end{pmatrix}.
\end{align*}
\end{remark}

\bibliographystyle{plain}
\bibliography{Lit_ArtikelJMP}

\end{document}